%% file: main.tex
\documentclass[runningheads]{llncs}
\usepackage[T1]{fontenc}
\usepackage{graphicx}
\input{packages}

\input{macros}
\usepackage{color}

\begin{document}
\title{Towards Actionable Strategy Certificates in Stochastic Parity Games}
\author{
	Christel Baier\inst{1}
	\and
	Diane Cauquil\inst{2}
	\and
	Calvin Chau\inst{1}
	\and
	Sascha Klüppelholz\inst{1}
	\and
	Anne-Kathrin Schmuck\inst{3}
}
\institute{
	$^1$ Technische Universität Dresden, Dresden, Germany\\
	\email{\{christel.baier,calvin.chau,sascha.klueppelholz\}@tu-dresden.de} \\
	$^2$ ENS Paris-Saclay, Paris, France\\
	\email{diane.cauquil@ens-paris-saclay.fr}\\
	$^3$ Max Planck Institute for Software Systems, Kaiserslautern, Germany\\
	\email{akschmuck@mpi-sws.org}\\
}
\authorrunning{Baier et al.}
\maketitle              %
\begin{abstract}
We propose a new approach for synthesizing large sets of winning strategies in stochastic parity games (2.5-player games) with quantitative objectives. Instead of computing a single, fully specified winning strategy, we introduce \ASCertsLong (\ASCerts) as a local and permissive representation of a large class of system player winning strategies. To this end, we extend known certificates for \emph{stochastic invariants} to the setting of games. Our certificates prove that synthesized strategies remain within a safe region of the game with probability at least $\lambda \in [0,1]$. As such, the certificates enhance the trustworthiness of synthesized strategies. The crux of our approach is to reinterpret and leverage the certificates as concise, local, and permissive representation of (possibly infinitely many) strategies. By carefully combining our certificates for stochastic invariants with strategy templates for almost-sure winning, we obtain a novel local representation of quantitatively winning strategies in stochastic parity games. This enables efficient synthesis, adaptation, and runtime strategy extraction, making \ASCerts well suited for logical control in uncertain and adversarial environments. We provide a proof-of-concept implementation and demonstrate the potential of applying \ASCerts in runtime adaptation on a case study.

\end{abstract}

\section{Introduction}
\input{sections/introduction}

\section{Preliminaries}
\label{section:prelims}
\input{sections/preliminaries}

\section{Overview and Problem Statement}
\label{section:overview}
\input{sections/overview}

 \section{Certificates for Stochastic Invariants}
 \label{section:certificates}
 \input{sections/certificates}

 \section{\ASCertsLong (\ASCerts)}
 \label{section:stochastic-template}
 \input{sections/stochastic_templates}

 \section{Experimental Evaluation}
 \label{section:experiments}
 \input{sections/evaluation}

 \section{Conclusion}
 \label{section:conclusion}
 \input{sections/conclusion}

 \input{credits}

\bibliographystyle{splncs04}
\bibliography{refs.bib}
 \newpage
 \appendix
 \section{Proofs for \Cref{section:certificates}}
 \input{sections/appendix-certificates}

 \section{Proofs for \Cref{section:stochastic-template}}
 \input{sections/appendix-stochastic-template}

 \section{Supplementary Material for \Cref{section:experiments}}
 \label{section:appendix-experiments}
 \input{sections/appendix-experiments}

\end{document}

%% file: packages.tex
\usepackage[english]{babel}

\usepackage{multirow}
\usepackage{booktabs}

\usepackage{amsmath}
\usepackage{graphicx}
\usepackage[colorlinks=true, allcolors=blue]{hyperref}
\usepackage{wasysym}
\usepackage{amssymb}
\usepackage[capitalise]{cleveref}
\usepackage{todonotes}
\usepackage{mathtools}
\usepackage[inline]{enumitem}
\usepackage{tcolorbox}
\usepackage{enumitem}%
\usepackage{apptools}
\AtAppendix{\counterwithin{lemma}{section}}

\usepackage{algorithm}
\usepackage[noend]{algpseudocode}

\usepackage{mdframed}

\usepackage{subcaption}

\usepackage{tikz}
\usetikzlibrary{arrows.meta, calc}
\usetikzlibrary{fit}

    \definecolor{paleblue}{rgb}{0.49,0.71,1.0}
    \definecolor{palered}{rgb}{1.0,0.25,0.25}
    
    \tikzset{statep1/.style={circle,draw,thick,fill=palered,minimum size=.8cm}}
      \tikzset{statep0/.style={rectangle,draw,thick,fill=paleblue,minimum size=.8cm, rounded corners}}
      \tikzstyle{line}=[-, >=latex, thick]
      \tikzstyle{arrow}=[->, >=latex, thick]

\usepackage{thmtools} 
\usepackage{thm-restate}

%% file: macros.tex
\newcommand{\union}{\cup}

\newcommand{\intersection}{\cap}
\newcommand{\prob}{\mathsf{Pr}}
\newcommand{\vect}[1]{\mathbf{#1}}
\newcommand{\reals}{\mathbb{R}}
\newcommand{\realsnn}{\reals_{\geq 0}}
\DeclareMathOperator{\suppOp}{supp}
\newcommand{\supp}[1]{\suppOp(#1)}

\newcommand{\naturals}{\mathbb{N}}

\newcommand{\card}[1]{{\lvert {#1} \rvert}}
\DeclareMathOperator{\distrOp}{Distr}
\newcommand{\distr}[1]{\distrOp(#1)}

\newcommand{\win}{\mathsf{Win}}

\newcommand{\expectation}[2]{\mathbb{E}^{#1}[#2]}

\newcommand{\Conj}{\bigwedge}

\newcommand{\eventually}{\lozenge}
\newcommand{\globally}{\square}

\newcommand{\states}{S}
\newcommand{\state}{s}

\newcommand{\actions}{Act}
\newcommand{\action}{a}
\newcommand{\init}{\bar{\state}}
\newcommand{\transRel}{\delta}
\newcommand{\transMat}{\vect{P}}
\newcommand{\exit}{\bot}

\newcommand{\SA}{\mathsf{En}}

\newcommand{\scheduler}{\sigma}
\newcommand{\schedulers}{\Sigma}

\newcommand{\plays}{\Pi}
\newcommand{\playsFin}{\Pi_{\mathsf{fin}}}
\newcommand{\play}{\pi}
\newcommand{\SApref}{\rho}

\newcommand{\sg}{\mathcal{G}}
\newcommand{\playerOne}{1}
\newcommand{\playerTwo}{2}

\newcommand{\player}{i}
\newcommand{\statesOne}{\states_\playerOne}
\newcommand{\statesTwo}{\states_\playerTwo}
\newcommand{\schedulersOne}{\schedulers_\playerOne}
\newcommand{\schedulersTwo}{\schedulers_\playerTwo}

\DeclareMathOperator{\srcOp}{src}
\newcommand{\src}[1]{\srcOp(#1)}

\newcommand{\p}[1]{\ensuremath{\text{Player}~#1}}
\newcommand{\pOne}{\ensuremath{\p{\playerOne}}}
\newcommand{\pTwo}{\ensuremath{\p{\playerTwo}}}

\newcommand{\subgame}[2]{\sg[#1,#2]}

\newcommand{\valueVector}{\vect{v}^*}

\newcommand{\coloring}{\mathbb{C}}
\newcommand{\parity}[1]{\mathsf{Parity}(#1)}
\newcommand{\even}[1]{{#1}_\mathsf{even}}
\newcommand{\odd}[1]{{#1}_\mathsf{odd}}

\newcommand{\template}{\Lambda}
\newcommand{\safetyTemplate}{\template_{\mathsf{Unsafe}}}
\newcommand{\liveTemplate}{\template_{\mathsf{Live}}}
\newcommand{\coliveTemplate}{\template_{\mathsf{Colive}}}

\newcommand{\almostParityTemplate}{\textsc{AsParityTemplate}}
\newcommand{\stochasticTemplate}{T}
\newcommand{\certs}{C}

\usepackage{xspace}
\newcommand{\ASCertLong}{Actionable Strategy Certificate\xspace}
\newcommand{\ASCertsLong}{Actionable Strategy Certificates\xspace}
\newcommand{\ASCert}{\textsf{ASCert}\xspace}
\newcommand{\ASCerts}{\textsf{ASCerts}\xspace}

\newcommand{\changes}[1]{#1}

\newcommand{\puddle}{%
	\tikz[baseline=-0.7ex]{%
		\fill[blue!75] (0,0) ellipse (0.18 and 0.10);%
		\fill[white, opacity=0.45] (-0.04,0.02) ellipse (0.06 and 0.03);%
	}%
}

\newcommand{\machine}{%
	\tikz[baseline=2.6ex]{%
\draw[gray, rounded corners=2pt, thick, fill=gray!60] (0.32,0.32) rectangle (0.68,0.68);
\node at (0.5, 0.5) {\includegraphics[scale=0.1]{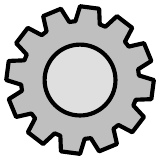}};
	}%
}

\newcommand{\station}{%
	\tikz[baseline=2.6ex, inner sep=0pt]{%
		\draw[green!50!black, rounded corners=2pt, thick, fill=green!70!black] (0.32,0.32) rectangle (0.68,0.68);
		\node at (0.5, 0.5) {\textsf{C}};
	}%
}

%% file: sections/introduction.tex
\emph{Games on graphs} provide a powerful and unifying abstraction for the analysis and synthesis of logical controllers in adversarial and uncertain environments. In this setting, the evolution of a system is modeled as a game between a system player, representing the controller to be synthesized, and one or more environment players, capturing uncontrollable or adversarial influences. The logical correctness of the system's evolution is then captured via temporal objectives -- such as safety, liveness, or fairness -- which are naturally expressed via $\omega$-regular conditions, leading to comprehensively studied classes of games on graphs \cite{fijalkow2025gamesgraphslogicautomata} such as parity games. Well-established algorithmic techniques, such as reactive synthesis \cite{finkbeiner2016synthesis}, are able to automatically compute a winning strategy for the system player in such games, which acts as a correct-by-construction logical controller for the modeled system evolution. As a result, games on graphs have become a cornerstone of formal methods for logical controller synthesis and verification of cyber-physical systems (CPS) \cite{KressGazitFainekosPappas2009,AlthoffBeltaReviewFMCEforAutonomousDriving,tabuada-2009-verification,belta2017formal,YIN2024100940,lindemann2025formal}.

However, while classical game-solving techniques yield strategies that are correct-by-design, they often fall short in providing suitable means for adaptability and refinement needed  in many CPS design settings, e.g., if new information about the environment becomes available or concrete action choices need to be deferred to other layers of the control stack. To address this limitation, \emph{permissive strategies} -- i.e., data structures that capture not only a single strategy but large sets of locally admissible actions -- have been advocated in various sub-fields of logical control design, e.g., supervisory control \cite{ramadge1989control}, runtime monitoring \cite{lee1999runtime,havelund2002synthesizing} or shielding \cite{konighofer2022correct,seshia2022toward}. Yet, almost all of these approaches are targeted towards safety certification, i.e., to ensure that something bad never happens. While safety is of utmost importance in CPS design, most systems also require liveness guarantees. For instance, warehouse robots need to eventually deliver packages, maintenance should be completed and customers should be served in a fair manner. Recent work therefore extends permissiveness to the full class of $\omega$-regular objectives via permissive strategy templates (PeSTels) \cite{anand_permissive_templates_2023}. Similar to their safety-predecessors, PeSTels emphasize a concise local representation which allows for adaptability and refinement both during synthesis and at runtime.  

While PeSTels have been applied successfully in the CPS context, e.g., for $\omega$-regular shielding of learned policies \cite{anand2025followstarsdynamicomegaregular} or for multi-level robot control \cite{10221705}, their computation typically assumes a deterministic model. Yet, in CPS scenarios stochasticity is usually unavoidable. Physical environments exhibit inherent uncertainty, learned models are imperfect, and interactions with other agents may be only probabilistically predictable. These considerations naturally lead to stochastic games, in particular 2.5-player games, where probabilistic transitions coexist with adversarial choices \cite{chatterjee_quantitative_games_2004}. In such games, correctness is often expressed not only qualitatively (almost-sure winning), but also \emph{quantitatively}, requiring that objectives are satisfied with probability at least a given threshold $\lambda \in [0, 1]$. Quantitative guarantees 
can often be more adequate than almost-sure correctness which may be either unrealistic or overly conservative. 

Unfortunately, the direct extension of PeSTels to quantitative winning strategies in stochastic parity games is not straightforward. Intuitively, PeSTels localize the information over frequencies of actions (never, finitely often or infinitely often) such that liveness is attained in the limit. This reasoning nicely transfers to almost-sure winning in stochastic games \cite{phalakarn_winning_2024}, but breaks in more general settings.  In fact, quantitative objectives inherently resist a purely local characterization, as the contribution of a local decision to a global probability threshold depends on complex interactions with future choices and stochastic outcomes.

In this paper, we thus take a different route in generalizing PeSTels to quantitative winning strategies in stochastic games by emphasizing simplicity, locality and adaptability over completeness. An important observation is that quantitative parity objectives can be decomposed into \emph{stochastic invariants} and \emph{almost-sure parity objectives} \cite{abate_quantitative_2025}. More precisely, if a strategy (i) ensures that it remains in a (safe) game region with probability at least $\lambda$ and (ii) within this region satisfies the parity objective (or leaves the region) almost surely, then it can be guaranteed that the parity objective is satisfied with probability at least $\lambda$. 

Based on this insight, we develop and combine representations for (i) safe strategies and (ii) almost-sure winning strategies into a concise, local, and permissive representation for quantitatively winning strategies in stochastic parity games. First, we extend certificates for stochastic invariants in Markov decision processes (MDPs) to stochastic games and leverage them as permissive representations of safe strategies. 
This leads to a novel \emph{actionable interpretation} of such certificates: strategic choices can act upon new information available in later stages of the synthesis process, e.g., when other components are fixed in distributed settings, low-level control choices have been made in multi-level control stacks, or refined knowledge about environment actions is available at runtime. %
We then carefully combine these certificates with strategy templates for almost-sure winning from \cite{phalakarn_winning_2024}, thereby obtaining our new formalism, which we coin \emph{\ASCertsLong (\ASCerts)}. %
We present a heuristic approach to synthesize \ASCerts, putting an emphasis on permissiveness, locality and simplicity. For the important class of \emph{globally optimal} strategies, i.e., strategies that behave optimally from every state, we show that our \ASCerts provide a \emph{relatively complete} characterization. That is, all strategies captured by our \ASCerts are guaranteed to be globally optimal. Conversely, there always exists a globally optimal strategy that is captured by our \ASCerts.
As a consequence, the synthesis procedure yields \ASCerts that are uniquely characterized. Finally, our synthesis techniques are implemented in a proof-of-concept tool. %

\smallskip
\noindent\textbf{Contribution.}
In summary, our contributions are as follows: %
\begin{itemize}
\item We extend certificates for stochastic invariants to stochastic games, which were previously only considered in the context of MDPs \cite{funke_farkas_2020,jantsch_dissertation_2022,chatterjee_fixed_point_2025} and Markov chains \cite{abate_quantitative_2025}. %
We then reinterpret them as permissive representations of strategies which comply with the given quantitative safety guarantees. Lastly, a heuristic approach for synthesizing such certificates which locally maximizes permissiveness in a tunable manner is presented.
\item Based on these certificates, we introduce \ASCerts, a new class of local, actionable certificates that characterize a large class of quantitatively winning strategies in stochastic parity games. To this end, we carefully combine our certificates for stochastic invariants with strategy templates for almost-sure winning \cite{anand_permissive_templates_2023,phalakarn_winning_2024}. Crucially, this combination requires a new formalization of randomized strategies which are extracted from strategy templates.
\item We then provide a synthesis method that computes \ASCerts for quantitative parity objectives. For the important class of \emph{globally optimal} strategies we show that our \ASCerts are \emph{relatively complete}. Further, the devised synthesis procedure for such optimal \ASCerts is simplified and yields \ASCerts that are uniquely characterized.
\item We provide an implementation of our techniques in a \emph{proof-of-concept} tool based on \textsc{Storm} \cite{HenselJKQV22}, \textsc{Prism-games} \cite{kwiatkowska_prism_games_2020}, \textsc{PeSTel} \cite{anand_permissive_templates_2023}, and \textsc{Mungojerrie} \cite{hahn_mungojerrie_2023}. We then demonstrate the \emph{actionability} aspect of \ASCerts in CPS settings by showcasing an application for runtime strategy adaptation. 
\end{itemize}
All omitted proofs are provided in the appendix. The implementation and experimental data are publicly available at \cite{baier_2026_22012734}.

\smallskip
\noindent\textbf{Related work.} %
The \emph{certificates} from \cite{funke_farkas_2020,jantsch_dissertation_2022,chatterjee_fixed_point_2025} for reachability probabilities and expected rewards in finite MDPs have been studied through the lens of \emph{certifying algorithms} \cite{certifying_algorithms_2011}, intended to increase the \emph{trustworthiness} of MDP verification. \emph{Martingale-based} certificates (e.g., \cite{ChakarovS13,TakisakaOUH18,TakisakaOUH21,ChatterjeeGMZ22,zikelic_rasm_2023,majumdar_supermatingale_2024,abate_quantitative_2025,HenzingerMSZ25}) have emerged as powerful tools for the verification and control of stochastic systems (with infinite state spaces). The decomposition of quantitative objectives into stochastic invariants and almost-sure objectives has been considered, e.g., in \cite{ChatterjeeNZ17,KretinskyM19,ChatterjeeGMZ22,abate_quantitative_2025,MajumdarS25}. The work in \cite{abate_quantitative_2025} showed that certificates for prefix-independent objectives in stochastic systems with general state spaces can be composed from certificates for stochastic invariants and certificates for almost-sure satisfaction. However, none of the aforementioned works consider certificates in the context of stochastic games and their usage as permissive strategy representations.

\emph{Strategy templates} have been introduced in \cite{anand_permissive_templates_2023,anand_templates_2024} for non-stochastic two-player games and have been extended to concurrent games (which require randomized strategies) \cite{concurrent_templates2026} and stochastic games \cite{phalakarn_winning_2024,phalakarn_templates_2025}. We note that \cite{phalakarn_winning_2024,phalakarn_templates_2025} do not provide an implementation of their techniques. Overall, only almost-sure, positive, or sure winning have been considered. Strategy templates for quantitative objectives have been studied in the setting of non-stochastic energy- and mean-payoff games \cite{anand_quantitative_2025}, but without any consideration of randomized strategies or optimality guarantees. The work in \cite{DBLP:conf/tacas/DragerFKPU14} presents multi-strategies for stochastic games that can describe multiple winning strategies. However, multi-strategies only capture memoryless strategies and do not consider $\omega$-regular objectives. With the same limitations, the work in \cite{DBLP:journals/corr/abs-2510-03481} considers multi-strategies for interval MDPs.

The model checker \textsc{Prism-games} \cite{kwiatkowska_prism_games_2020} supports the verification of safety and reachability properties in stochastic games, but parity and $\omega$-regular objectives, in general, are currently not supported. \textsc{Mungojerrie} \cite{hahn_mungojerrie_2023} can solve stochastic parity games via reinforcement learning, and provides optimality guarantees in the limit. \textsc{Gist} \cite{chatterjee_gist_2010} supports stochastic games, but does not appear to be available anymore. The \textsc{Storm} \cite{HenselJKQV22} model checker currently supports parsing but not solving stochastic games. We refer to \cite[Sec.\ 12]{qcomp2024} for a tool comparison.
Building on this limited tool support, we develop a synthesis chain for stochastic parity games which integrates the \textsc{PeSTel} tool \cite{anand_permissive_templates_2023} for strategy template computation in non-stochastic parity games with \textsc{Prism-games}, \textsc{Mungojerrie} and \textsc{Storm}. %

%% file: sections/preliminaries.tex
\noindent\textbf{Notation.}
We define $[i;j] = \{n \in \naturals \mid i \leq n \leq j \}$ and $[j] = [1;j]$ for $i,j \in \naturals$. We denote the set of even (odd) numbers in $[i;j]$ by $\even{[i;j]}$ ($\odd{[i;j]}$). Vectors are written in boldface, e.g., $\vect{x}$. Let $\states = \{\state_0, \dots, \state_n \}$ be a finite set. We write $\vect{x} \in \reals^\states$ instead of $\vect{x} \in \reals^\card{\states}$ and $\vect{x}(\state_i)$ to denote its $i$th entry. For $A \subseteq \states$, we define $\vect{1}_{A} \in \{0, 1\}^{\states}$ by $\vect{1}_{A}(s) = 1 \iff s \in A$. The \emph{support} of $\vect{x} \in \realsnn^\states$ is defined by $\supp{\vect{x}} = \{\state \in \states \mid \vect{x}(\state) > 0\}$. We call $\vect{x} \in [0,1]^\states$ a \emph{probability distribution} over $\states$ if $\sum_{\state \in \states} \vect{x}(\state) = 1$. The set of probability distributions over $S$ is denoted by $\distr{\states}$. We use $\states^*$, $\states^+$, and $\states^\omega$ to denote the sets of all finite, non-empty finite, and infinite words over $\states$, respectively.

\medskip

\noindent\textbf{Stochastic Games.} A \emph{turn-based stochastic game} (SG) \cite[Ch.~7]{fijalkow2025gamesgraphslogicautomata} $\sg$ is a tuple $(\states, \statesOne, \statesTwo, \init, \actions, \transRel)$ where $\states$ is a finite set of states, $(\statesOne, \statesTwo)$ a partition of $\states$, $\statesOne$ and $\statesTwo$ the states of $\pOne$ and $\pTwo$, respectively, $\init$ the initial state, $\actions$ a finite set of actions and $\transRel \colon \states \times \actions \rightharpoonup \distr{\states}$ a partial transition function. We omit the initial state $\init$ from the definition when it is not relevant. We often write $\transRel(\state, \action, \state')$ instead of $\transRel(\state, \action)(\state')$. For $I \subseteq \states$, we define $\transRel(\state, \action, I) = \sum_{\state' \in I} \transRel(\state, \action, s')$. An action~$\action$ is \emph{enabled} in state~$\state$ if $\transRel(\state, \action)$ is defined. We denote the enabled actions in $\state$ by $\actions(\state)$ and assume that $\actions(\state) \neq \emptyset$ for all states. The \emph{enabled state-action pairs} are denoted by $\SA \subseteq \states \times \actions$.
An infinite \emph{play} $\play$ in $\sg$ is a state-action sequence $\play = \state_0 \action_0 \state_1 \action_1 \ldots \in \SA^{\omega}$ such that $\action_n \in \actions(\state_n)$ and $\transRel(\state_n, \action_n, \state_{n+1}) > 0$ for all $n \geq 0$. We write $\play_{\leq n}$ for the prefix $\state_0 \action_0 \ldots \action_{n-1} \state_n$ and $\play_{=n}$ for $\state_n$. Finite plays $\play \in \SA^* \states$ are defined analogously, and we often write $\SApref \state$ to denote them, where $\SApref \in \SA^*$ and $\state \in \states$. The sets of infinite and finite plays starting in $\state$ are denoted by $\plays(\sg, \state)$ and $\playsFin(\sg, \state)$, respectively, and we omit $\state$ to denote plays starting from any state and drop $\sg$ when clear from context.

A \emph{strategy} for Player $\player \in \{\playerOne, \playerTwo\}$ is a function $\scheduler_i \colon \SA^* \states_\player \to \distr{\actions}$ such that $\supp{\scheduler_i(\SApref \state)} \subseteq \actions(\state)$ for all $\SApref \state \in \SA^* \states_\player$. A strategy $\scheduler_i$ is \emph{memoryless} if $\scheduler_i(\SApref \state) = \scheduler_i(\SApref'\state)$ for all plays $\SApref \state$ and $\SApref' \state$, in which case we treat it as a function $\scheduler_i \colon \states_i \to \distr{\actions}$. A strategy is \emph{pure} if $\scheduler_i(\play)$ is Dirac for all plays $\play$. We denote the set of all Player $\player$ strategies by $\schedulers_\player$. A play $\play = \state_0 \action_0 \ldots$ is $\scheduler_i$-\emph{compliant} if $\scheduler_i(\state_0 \ldots \state_n)(\action_n) > 0$ for all $n \geq 0$. The set of infinite and finite $\scheduler_i$-compliant plays is denoted by $\plays^{\scheduler_i}$ and $\playsFin^{\scheduler_i}$, respectively. For a strategy profile $(\scheduler_\playerOne, \scheduler_\playerTwo)$ and $\state \in \states$, we consider the standard probability measure $\prob_{\sg, \state}^{\scheduler_\playerOne, \scheduler_\playerTwo}$ on infinite plays starting from $\state$ (see, e.g., \cite[Ch. 10]{baier_principles_2008}).

We say that an SG $\sg$ is \emph{non-stochastic} if $\transRel(\state, \action)$ is \emph{Dirac} for all $(\state, \action) \in \SA$. 
A \emph{Markov Decision Process} (MDP) is an SG where $\statesTwo = \emptyset$.%

\medskip

\noindent\textbf{Objectives.} Given an SG $\sg$ with state-action pairs $\SA$, we consider \emph{objectives} $\win\subseteq \SA^\omega$ that can be expressed as \emph{LTL formulae} $\Phi$ (see, e.g., \cite[Ch.~5]{baier_principles_2008}) whose atomic propositions are subsets of $\SA$, i.e., $\play\in\win$ iff $\play\vDash\Phi$. 
We write $(\sg, \Phi)$ to denote an SG $\sg$ with objective\footnote{We often slightly abuse notation and use $\Phi$ to refer to both $\win$ and $\Phi$. } $\Phi$. For a given probability threshold $\lambda \in [0, 1]$, a strategy $\scheduler_1$ for $\pOne$ is \emph{quantitatively winning} in $(\sg, \Phi)$ if $\inf_{\scheduler_2 \in \schedulers_2} \prob_{\sg, \init}^{\scheduler_1, \scheduler_2}(\Phi) \geq \lambda$, i.e., the probability of a $\scheduler_1$-compliant play to satisfy $\Phi$ is at least $\lambda$ independently of $\scheduler_2$. A $\pOne$ strategy $\scheduler_1$ is \emph{almost-sure winning} if $\lambda = 1$ and \emph{sure winning} if $\plays^{\scheduler_1}(\sg, \init)\subseteq\win$. %
We define the almost-sure winning region $\mathcal{W}_1$ as the set of states from which $\pOne$ has an almost-sure winning strategy.
The \emph{value vector} $\valueVector_{\sg, \Phi} \in [0,1]^{\states}$ of $(\sg, \Phi)$ is defined by
$\valueVector_{\sg, \Phi}(\state) = \sup_{\scheduler_1 \in \schedulers_1} \inf_{\scheduler_2 \in \schedulers_2} \prob_{\sg, \state}^{\scheduler_1, \scheduler_2}(\Phi)$ for all $\state\in\states$.
A strategy $\scheduler_1$ for $\pOne$ is \emph{globally optimal} if
$\inf_{\scheduler_2 \in \schedulers_2} \prob_{\sg, \state}^{\scheduler_1, \scheduler_2}(\Phi) = \valueVector_{\sg, \Phi}(\state)$ for all $\state \in \states$.

We consider reachability and invariant objectives $\eventually I$ and $\globally I$, respectively, where $I \subseteq \states$, requiring that a play eventually visits and always remains in $I$, respectively. Further, \emph{parity objectives} are LTL formulae of the form:
\begin{align}
	\textstyle\parity{\coloring} \coloneqq \bigwedge_{k \in \odd{[0;d]}} \left( \globally \eventually C^k \implies \bigvee_{j \in \even{[k+1; d]}} \globally \eventually C^j \right),
\end{align}
where $\coloring \colon \states \to [0;d]$ is a \emph{coloring} and $C^j = \{(\state, \action) \in \SA \mid \mathbb{C}(\state)  = j \}$. W.l.o.g. we always assume that $d$ is even. Intuitively, parity objectives require the largest \emph{color} that is visited infinitely often to be even.

\medskip

\noindent\textbf{Strategy Templates.} Given an SG $\sg$, a \emph{strategy template} \cite{anand_permissive_templates_2023} captures a (possibly infinite) set of $\pOne$ strategies via three types of local conditions over state-action pairs. Formally, a strategy template is a tuple $(U,D,\mathcal{H})$ consisting of \emph{unsafe} state-action pairs $U \subseteq \SA$ -- which may never be taken, a set of \emph{co-live} state-action pairs $D \subseteq \SA$ -- which may be taken only finitely often, and a set of \emph{live-groups} $\mathcal{H} \subseteq 2^\SA$ -- which require that at least one action of the group is taken infinitely often if one of the source states in the group is seen infinitely often. 
A strategy template $(U,D,\mathcal{H})$ represents an LTL formula %
\begin{subequations}\label{equ:strategytemplate}
	\begin{align}
		\template \coloneqq \safetyTemplate(U) &\land \coliveTemplate(D) \land \liveTemplate(\mathcal{H})\notag\\
		\text{s.t.\ }~\safetyTemplate(U) &\textstyle\coloneqq \globally \bigwedge_{(\state, \action) \in U} \neg (\state, \action),\\
		\coliveTemplate(D) &\textstyle\coloneqq \bigwedge_{(\state, \action) \in D} \eventually \globally \neg (\state, \action),\text{ and}\\
		\liveTemplate(\mathcal{H}) &\textstyle\coloneqq \Conj_{H \in \mathcal{H}} \globally \eventually \src{H} \Rightarrow \globally \eventually H, 
	\end{align}
\end{subequations}
where $\src{H} = \{\state \in \states \mid \exists \action \centerdot (\state, \action) \in H\}$.
As in \cite{anand_permissive_templates_2023,phalakarn_winning_2024,phalakarn_templates_2025}, a $\pOne$ strategy is said to \emph{follow} a strategy template\footnote{We slightly abuse notation and use $\template$ to refer both to the tuple $(U,D,\mathcal{H})$ and \cref{equ:strategytemplate}.} $\template$ almost surely if $\pOne$ is almost surely winning in $(\sg,\template)$ from every state. A strategy template $\template$ is almost surely winning in $(\sg,\Phi)$ from state $s$ if every $\pOne$ strategy following $\template$ is almost surely winning in $(\sg,\Phi)$ from state $s$. A template is almost surely winning if it is almost surely winning from every state in the almost-sure winning region $\mathcal{W}_1$.%

Existing algorithms \cite{anand_permissive_templates_2023,phalakarn_winning_2024,phalakarn_templates_2025} compute an almost surely winning strategy
template for a stochastic parity game with the same time complexity as classical strategy synthesis. We denote the algorithm from \cite{phalakarn_winning_2024,phalakarn_templates_2025} which computes an almost-sure winning
strategy template in a stochastic parity game by\linebreak $\almostParityTemplate(\sg, \Phi)$. 
The work in \cite{phalakarn_templates_2025} does not provide an implementation of this algorithm but suggests to reduce the stochastic parity game to a non-stochastic parity game (using the gadgets from \cite{chatterjee_quantitative_games_2004}) and use the \textsc{PeSTel} tool \cite{anand_permissive_templates_2023} to solve the resulting game. We adopt this approach in our experiments.

%% file: sections/overview.tex
\begin{figure}[t]
\centering
\begin{subfigure}[t]{.4\textwidth}
  \centering
  \includegraphics[scale=0.5]{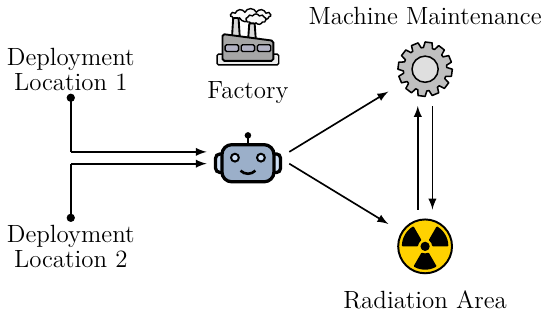}

  \label{fig:sub1}
\end{subfigure}%
\hfill
\begin{subfigure}[t]{.59\textwidth}
  \centering
  \includegraphics[scale=0.7]{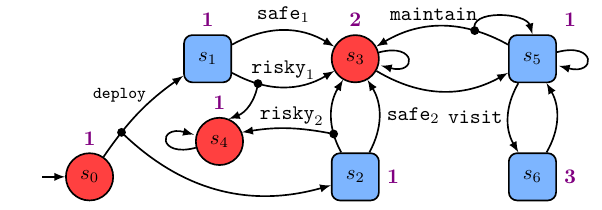}
  \label{fig:sub2}
\end{subfigure}
\vspace{-0.2cm}
\caption{Illustration of the motivating example (left) and the underlying stochastic game (right). $\pOne$ ($\pTwo$) states are shown as boxes (circles). The coloring $\coloring$ is shown next to each state. Transition probabilities are omitted and probabilistic transitions (indicated by a dot) are uniformly distributed.}
\vspace{-0.4cm}
\label{fig:running-example}
\end{figure}
Before explaining our contributions in more detail, we illustrate the synthesis problem under consideration by a simple motivating example.

\smallskip
\noindent\textbf{A Motivating Example.} We consider a robot that needs to maintain machines at a factory. The setting is illustrated in \cref{fig:running-example} (left). Initially, the robot is randomly deployed at location $1$ or $2$, and first needs to get to the factory. It can choose between the safe path, guaranteeing that it reaches the factory, and the risky path, where it might get permanently stuck. Once at the factory, the robot needs to maintain the machines infinitely often and is either kept there or sent away. Then, the robot may decide to visit an area where radiation levels are high and which it may only visit finitely often (otherwise it sustains permanent damage). It can return from the radiation area, but might get stuck.

The \emph{stochastic parity game} $(\sg, \Phi)$ shown in \cref{fig:running-example} (right) models the example. The robot is controlled by $\pOne$ and the environment corresponds to $\pTwo$. The deployment phase is modeled by the states $\state_0$, $\state_1$, $\state_2$, and $\state_4$. In $\state_1$ (or $\state_2$) the robot can decide to choose the safe or risky path leading to the factory at state $\state_3$. When taking the risky path, with probability $0.5$ the robot gets permanently stuck at $\state_4$. If it arrives at $\state_3$, it needs to ensure that $\state_3$ is visited infinitely often (i.e., maintaining the machines), modeled by assigning $\state_3$ an even color, i.e., $\coloring(\state_3) = 2$. The radiation area is modeled by state $\state_6$, which the robot may only visit finitely often, hence it has an odd color, i.e., $\coloring(\state_6) = 3$. The parity objective $\Phi = \parity{\coloring}$ then requires that the largest color that is seen infinitely often is \emph{even}. Observe that $\pOne$ can satisfy the objective \emph{almost surely}, but not \emph{surely}, since there are plays that remain in $\state_5$ (which have measure $0$). %

Given the stochastic parity game $(\sg,\Phi)$ in \cref{fig:running-example} (right), one can (via standard techniques) compute a strategy $\sigma_1$ which is quantitatively winning in $(\sg,\Phi)$ for some probability threshold $\lambda\in[0,1]$. In this example, an optimal (almost-sure winning) strategy is, for instance, given by a pure memoryless strategy that always picks $\texttt{safe}$ in $\state_1$ and $\state_2$, and chooses $\texttt{maintain}$ in $\state_5$. Instead, if we can tolerate that the machines are only maintained with probability $\geq 0.75$, we can also consider strategies that choose $\texttt{risky}$ with probability $0.5$ in $\state_1$ and $\state_2$.

Crucially, however, these single (fully-specified) strategies might not always be applicable, as the choice of actions available to the robot might change at runtime, e.g., the safe path might be blocked, or a low-level control layer enforces a different choice, e.g., guides the robot to the radiation area to avoid a human operator entering the maintenance zone. Further, the risky path could be shorter and at runtime the robot might need to arrive at the factory as soon as possible.
While the effect of these influences could be modeled into the game, they are often not known at design time or increase the size of the model significantly while only occurring rarely in practice. To make strategies robustly tolerate such rare influences, we would like to retain local information about what actually makes a strategy quantitatively winning, to adapt local action choices dynamically.

\smallskip
\noindent\textbf{Challenges and Problem Statement.} We propose \ASCertsLong (\ASCerts) as a solution to this problem. The challenge in the design of \ASCerts is to strike a careful balance between completeness and simplicity. On the one hand, we want to capture rich classes of strategies, so as not to exclude possible solutions in the design of CPS. On the other hand, we need \ASCerts to be simple and local to allow for a straightforward extraction of strategies that enables runtime adaptations. Lastly, we also require a practical synthesis approach for \ASCerts. Our problem statement can thus be summarized as follows. %
\begin{mdframed}
\textbf{Problem Statement:} Given a stochastic parity game $(\sg, \Phi)$, capture strategies $\scheduler_1$ for $\pOne$ (possibly infinite memory and/or randomized) that are quantitatively winning with probability threshold $\lambda \in [0, 1]$ in a \emph{concise}, \emph{local}, and \emph{permissive} formalism.
\end{mdframed}

Within this paper we solve the stated problem via a combination of quantitative certificates and qualitative strategy templates as explained next.

\smallskip
\noindent\textbf{Stochastic Invariants.} The main insight that guides our development of such `robust strategies' is the use of certificates for stochastic invariants.
Suppose we want to describe strategies $\scheduler_1$ for $\pOne$ that enable the robot to fulfill its task with probability $\geq 0.75$. Note that once the robot reaches $\state_4$ it can no longer satisfy its objective. Thus, a winning strategy $\scheduler_1$ needs to ensure that $\inf_{\scheduler_2 \in \schedulersTwo}\prob_{\sg, \init}^{\scheduler_1, \scheduler_2}(\globally I) \geq 0.75$, where $I = \states \setminus \{\state_4\}$. %
On its own such invariants $I$ are not yet \emph{actionable}, as it is not clear how a strategy that remains in $I$ can be extracted. Such strategic choices cannot be made for each state locally and independently, since, e.g., taking the risky path both in $\state_1$ and $\state_2$ with probability $1$ leads to failure (when $\lambda = 0.75$). To tackle the complex interplay of such decisions, we make use of \emph{certificates} for stochastic invariants. We formally present them in \cref{section:certificates}, but for now it is helpful to think of a certificate as a mapping that assigns a real-valued \emph{rank} to each state. We then require that the strategy $\scheduler_1$ ensures that this rank \emph{does not increase in expectation}. For instance, a certificate for the stochastic invariant $I$ in our example is given by:
\begin{align}
\vect{x} = (\state_0 \mapsto 0.25,\; \state_1 \mapsto 0.25,\; \state_2 \mapsto 0.25, \state_3 \mapsto 0,\; \state_4 \mapsto 1,\; \state_5 \mapsto 0,\; \state_6 \mapsto 0). \label{eq:example-certificate}
\end{align}
Then, it is sufficient that whenever $\scheduler_1$ arrives at $\state_1$ (or $\state_2$) the following holds:
\begin{equation}
\begin{aligned}
\vect{x}(\state_1) = 0.25 \geq p_\texttt{safe} \cdot \vect{x}(\state_3) + p_\texttt{risky} \cdot \left(0.5 {\cdot} \vect{x}(\state_3) + 0.5 {\cdot} \vect{x}(\state_4) \right) = p_\texttt{risky} \cdot 0.5 \label{eq:example-local-constraint}
\end{aligned}
\end{equation}
where $p_{\texttt{risky}}$ ($p_{\texttt{safe}}$) is the probability with which $\scheduler_1$ chooses the risky (safe) path. We also say that $\scheduler_1$ \emph{follows} the certificate $\vect{x}$. In fact, $\pOne$ may choose \emph{any strategy} that follows $\vect{x}$ (i.e., satisfies these \emph{local constraints}), e.g., pick any $p_\texttt{risky} \in [0, 0.5]$ in our example. Notably, such certificates allow us to reason about arbitrary strategies, e.g., even infinite-memory strategies.

\smallskip
\noindent\textbf{\ASCertsLong (\ASCerts).} Provided that we have an invariant $I$ and certificate $\vect{x}$ as before, we know that any strategy $\scheduler_1$ that follows $\vect{x}$ stays in $I$ with probability $\geq 0.75$. However, this is not enough to ensure that $\Phi$ is satisfied. For instance, a strategy following $\vect{x}$ is still allowed to choose $\texttt{visit}$ infinitely often. Hence, we need to ensure that certain actions are only taken finitely often and that progress towards the objective is made. Here, the strategy templates for almost-sure winning come into play. In this example, the template will prescribe that $\scheduler_1$ chooses \texttt{maintain} infinitely often whenever it visits $\state_5$ infinitely often. Additionally, it requires that $\texttt{visit}$ is only chosen finitely often. Our \ASCerts then combine the certificate and strategy templates to capture winning strategies (\Cref{def:stochastic-template-follow}). As a key technical contribution, we provide an instructive characterization with which strategies can be directly extracted from \ASCerts (\Cref{sec:extraction}), enabling, e.g., adaptation at runtime. %

\smallskip
\noindent\textbf{\ASCerts Synthesis.}
So far, we have assumed to be given an invariant $I$ and certificate $\vect{x}$. In \cref{section:certificates} we describe an approach for computing permissive certificates and \cref{section:stochastic-template} discusses how invariants can be determined. This provides the basis for the synthesis of our \ASCerts, shown in \cref{fig:ascert-synthesis-overview}. Starting from a quantitatively winning strategy $\scheduler_1$ obtained via strategy synthesis, we show that both the invariant $I$ and certificate $\vect{x}$ can be computed. Then, both $I$ and $\vect{x}$ are used to construct a subgame, for which we invoke \textsc{PeSTel} \cite{anand_permissive_templates_2023} to compute the almost-sure winning strategy template w.r.t.\ a slightly modified parity objective. For the important class of globally optimal strategies we show that the \ASCert synthesis is simplified as it does not depend on a specific strategy $\scheduler_1$.

\begin{figure}[t]
\centering
\includegraphics[scale=0.7]{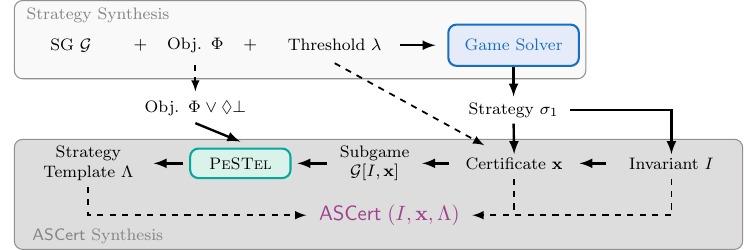}
\caption{Overview of our \ASCert synthesis approach.}
\label{fig:ascert-synthesis-overview}
\vspace{-0.4cm}
\end{figure}

%% file: sections/certificates.tex
This section presents \emph{certificates} for \emph{stochastic invariants} in SGs that will pave the way for \ASCerts, formally introduced in \Cref{section:stochastic-template}. The crux of our approach is to reinterpret certificates as \emph{concise strategy representations}. We then also discuss a heuristic approach for synthesizing \emph{permissive} certificates.

\subsection{Actionable Certificates in Stochastic Games}
In the following, we provide a new characterization of certificates for stochastic invariants in SGs. Intuitively, certificates are solutions to constraint systems and prove that a strategy remains in a region $I \subseteq \states$ with a desired probability $\lambda$.
\begin{restatable}{theorem}{stochasticInvariantCertificates}
Let $\sg = (\states, \statesOne, \statesTwo, \init, \actions, \transRel)$ be an SG, $I \subseteq \states$, $\lambda \in [0, 1]$ and $\scheduler_1 \in \schedulersOne$ be a $\pOne$ strategy. Suppose there exists $\vect{x} \in [0,1]^{\states}$ such that:
\begin{subequations}\label{eq:certificates}
\begin{align}
 \vect{x}(\init) &\leq 1 - \lambda, && \label{eq:threshold} \\
\vect{x}(\state) &\geq \sum_{\state' \in \states} \transRel(\state, \action, \state') \cdot \vect{x}(\state') && \forall \state \in I \intersection \statesTwo, \forall \action \in \actions(\state), \label{eq:opponent}\\
\vect{x}(\state) &= 1 && \forall \state \in \states \setminus I, \label{eq:losing-states}\\
 \vect{x}(\state) &\geq \sum_{\action \in \actions(\state)} \sum_{\state' \in \states} \scheduler_1(\SApref \state)(\action) {\cdot} \transRel(\state, \action, \state') {\cdot} \vect{x}(\state') && \forall \state \in I \intersection \states_1,\; \forall \SApref \state \in \playsFin^{\scheduler_1}.\label{eq:strategy}
\end{align}
\end{subequations}
Then, we have $\inf_{\scheduler_2 \in \schedulersTwo} \prob_{\sg, \init}^{\scheduler_1, \scheduler_2}(\globally I) \geq \lambda$, i.e., $\scheduler_1$ is quantitatively winning in the safety game $(\sg,\globally I)$ with probability at least $\lambda$.
\label{theorem:cert-invariant}
\end{restatable}
\noindent A vector $\vect{x}$ that satisfies the constraints in \eqref{eq:certificates} bounds the probability of reaching $\states \setminus I$ under $\scheduler_1$ (and all $\pTwo$ strategies) from above and \emph{certifies} that $\prob_{\sg, \init}^{\scheduler_1, \scheduler_2}(\globally I) \geq \lambda$ for all $\scheduler_2 \in \schedulersTwo$ holds. Additionally, it can be interpreted as a \emph{ranking function} that is required to be non-increasing in expectation.

We call a vector $\vect{x}$ that satisfies \eqref{eq:threshold}--\eqref{eq:losing-states} a \emph{certificate} for $I$ and $\lambda$ and denote the set of certificates by $\certs_{\geq \lambda}(\sg, I)$. Further, $\vect{x} \in \certs_{\geq \lambda}(\sg, I)$ is a certificate for a $\pOne$ strategy $\scheduler_1$ if it additionally satisfies \eqref{eq:strategy}. Note that the definition of $\certs_{\geq \lambda}(\sg, I)$ does not depend on a $\pOne$ strategy. We now present a shift in perspective and consider the strategies that are \emph{represented} by a certificate.%
\begin{definition}\label{def:follow-cert}
	Let $\vect{x} \in \certs_{\geq \lambda}(\sg, I)$. A $\pOne$ strategy $\scheduler_1$ is said to \emph{follow} $\vect{x}$ if
	\begin{align}
		&\vect{x}(\state) \geq \sum_{\action \in \actions(\state)} \sum_{\state' \in \states} \scheduler_1(\SApref \state)(\action) \cdot \transRel(\state, \action, \state') \cdot \vect{x}(\state')  & \forall \state \in I \intersection \states_1,\; \forall \SApref \state \in \playsFin^{\scheduler_1}. \label{eq:follow}
	\end{align}
\end{definition}
Intuitively, a strategy \emph{follows} a certificate $\vect{x} \in \certs_{\geq \lambda}(\sg, I)$ if it can ensure that $\vect{x}$ is non-increasing in expectation. We note that the constraints \eqref{eq:follow} and \eqref{eq:strategy} coincide and that $\scheduler_1$ follows $\vect{x}$ if and only if $\vect{x}$ is a certificate for $\scheduler_1$. Thanks to \Cref{theorem:cert-invariant} we then directly obtain soundness.
\begin{corollary}
	Let $\vect{x} \in \certs_{\geq \lambda}(\sg, I)$. If $\scheduler_1$ \emph{follows} $\vect{x}$ then $\inf_{\scheduler_2 \in \schedulersTwo}\prob_{\sg, \init}^{\scheduler_1, \scheduler_2}(\globally I) \geq \lambda$.
\end{corollary}
A certificate represents a (potentially infinite) set of $\pOne$ strategies that ensure that the play stays in $I$ with probability at least $\lambda$. A useful property we will exploit is that condition \eqref{eq:follow} is \emph{local}, i.e., a strategy can follow a certificate in state $\state$ by only considering $\state$ and its immediate successors. Thus, the decision of $\scheduler_1$ in state $\state$ does not depend on its decisions in the other states.

For memoryless strategies the certificates are complete, i.e., for a memoryless strategy $\scheduler_1$ there always exists a certificate $\vect{x}$ that is followed by $\scheduler_1$.
\begin{lemma}[Completeness for Memoryless Strategies] \label{lemma:complete-mr}
If a memoryless strategy $\scheduler_1$ satisfies $\inf_{\scheduler_2 \in \schedulersTwo} \prob_{\sg, \init}^{\scheduler_1, \scheduler_2}(\globally I) \geq \lambda$, then there exists a certificate $\vect{x}$ for $I$ and $\lambda$ (i.e., $\vect{x} \in \certs_{\geq \lambda}(\sg, I)$) such that $\scheduler_1$ follows $\vect{x}$.
\end{lemma}
\begin{proof}
    Applying $\scheduler_1$ to $\sg$ yields a finite MDP, for which the completeness of certificates is well-known, see, e.g., \cite{jantsch_dissertation_2022} and \cite{chatterjee_fixed_point_2025}. Particularly, constraint \eqref{eq:strategy} simplifies to a quantification over $\state \in I \intersection \states_1$. \qed
\end{proof}
In the context of certifying algorithms \cite{certifying_algorithms_2011}, certificates are returned as part of an algorithm's output to increase its trustworthiness. That is, the certificate serves as an easy-to-check proof of correctness for the algorithm's result.
In the setting of strategy synthesis procedures, a synthesized (finite-memory) strategy $\scheduler_1$ can be augmented with a certificate $\vect{x} \in \certs_{\geq \lambda}(\sg, I)$ that is followed by $\scheduler_1$ to certify that $\scheduler_1$ remains in some (safe) region $I$ with probability at least $\lambda$. Checking a certificate then just amounts to checking whether constraints \eqref{eq:certificates} %
are satisfied.
\begin{remark}
	In the setting of MDPs, the work in \cite{chatterjee_fixed_point_2025} also represents strategies via certificates. While their certificates only describe pure strategies, ours also represent \emph{randomized} strategies. \changes{Further, since we capture strategies that potentially use infinite memory, we need to reason via \emph{supermartingales} in the proof of \Cref{theorem:cert-invariant}, while \cite{chatterjee_fixed_point_2025,funke_farkas_2020} can rely on well-known fixed-point characterizations.}
\end{remark}
\begin{remark}
	Our certificates for $I$, $\lambda$ and $\scheduler_1$ are similar to the certificates for Markov chains with general state spaces from \cite[Theorem 2]{abate_quantitative_2025}. The key difference lies in \eqref{eq:opponent} where we need to treat the nondeterminism arising from $\pTwo$. For finite MDPs similar certificates were already considered in \cite{jantsch_dissertation_2022,funke_farkas_2020} and \cite{chatterjee_fixed_point_2025}.
\end{remark}

\subsection{Permissive Certificate Synthesis}
\label{subsection:permissive-certificates}
We now discuss the synthesis of our certificates. For this purpose, we assume that we are given $I \subseteq \states$, $\lambda \in [0,1]$ and a memoryless strategy $\scheduler_1$ for $\pOne$. As we shall see later, these assumptions will be compatible with the overall synthesis approach of our \ASCerts (see \Cref{section:stochastic-template} and \Cref{fig:ascert-synthesis-overview}). First note that for a fixed set $I \subseteq \states$ and threshold $\lambda$ the certificate is generally not uniquely determined and we have a certain level of freedom.
\begin{example}\label{example:not-permissive}
We revisit our motivating example from \Cref{section:overview} involving the maintenance robot. Recall that we have considered the certificate in \eqref{eq:example-certificate}, resulting in the constraint \eqref{eq:example-local-constraint} which allowed $\pOne$ to choose the \emph{risky} path with probability $\leq 0.5$. Now consider the following certificate:
\begin{align}
\vect{x}' = (\state_0 \mapsto 0,\; \state_1 \mapsto 0,\; \state_2 \mapsto 0, \state_3 \mapsto 0,\; \state_4 \mapsto 1,\; \state_5 \mapsto 0,\; \state_6 \mapsto 0).
\end{align}
Since $\vect{x}'(\state_1) = 0$ and $\pOne$ needs to ensure that $\vect{x}'$ does not increase in expectation, $\vect{x}'$ only enables $\pOne$ to choose the \emph{safe} path from $\state_1$.
\end{example}
The previous example shows that the provided flexibility depends on the choice of the certificate. Towards a procedure for synthesizing certificates that provide ``much'' flexibility, we define a notion of \emph{permissiveness} for our certificates.
\begin{definition}[$\varepsilon$-Permissiveness]
	Let $\vect{x} \in \certs_{\geq \lambda}(\sg, I)$ be a certificate, $\state \in I \intersection \states_1$, and $\varepsilon \in [0, 1]$. We say that $\vect{x}$ is $\varepsilon$-\emph{permissive} in $\state$ if 
	\begin{align}
		\textstyle\vect{x}(\state) \geq \left( \sum_{\state' \in \states} \transRel(\state, \action, \state') \cdot \vect{x}(\state') \right) - (1 - \varepsilon) \qquad \forall \action \in \actions(\state). \label{eq:most-permissive}
	\end{align}
	We say that $\vect{x}$ is \emph{most permissive} in $\state$ if it is $1$-permissive in $\state$.
\end{definition}
Intuitively, a certificate $\vect{x}$ that is most permissive in state~$\state$ enables $\pOne$ to choose any action (with arbitrary probability), as it is guaranteed that $\vect{x}$ will not increase in expectation under any action.
Based on $\varepsilon$-permissiveness, we propose the following heuristic approach for finding permissive certificates. We construct a linear program (LP) with the constraints \eqref{eq:threshold}--\eqref{eq:strategy} w.r.t. $I$, $\lambda$ and $\scheduler_1$. Additionally, we introduce a variable $\varepsilon_{\state}$ and the corresponding constraint \eqref{eq:most-permissive} for each state $\state \in \states_1$. Then we determine a certificate $\vect{x}$ that maximizes $\sum_{\state \in \states_1} \varepsilon_\state$. Thereby, we try to maximize the \emph{overall permissiveness} of the certificate. A natural variation of this heuristic is to instead only maximize permissiveness in a subset of states where flexibility is needed the most. We describe the LP in more detail in \Cref{subsection:permissive-lp}. In \Cref{section:experiments} we demonstrate that our approach yields certificates that grant greater flexibility, enabling effective runtime adaptation. %

Finally, observe that there is an inherent trade-off between maximizing the overall permissiveness and guaranteeing high safety assurances (i.e., high probability thresholds $\lambda$). For instance, choosing $\vect{x} = \vect{1}$ maximizes the overall permissiveness and ensures that \emph{all} states are most permissive, but does not provide any safety assurances (cf.\ \cref{eq:threshold}). Conversely, as seen in \Cref{example:not-permissive}, providing high assurances reduces the permissiveness.

\begin{example}
The certificate $\vect{x}$ in \cref{eq:example-certificate} for our motivating example from \Cref{section:overview} maximizes the overall permissiveness w.r.t. $I = \states \setminus \{\state_4\}$ and $\lambda=0.75$ and is based on the optimal memoryless strategy $\scheduler_1$ that never picks a risky path (see \cref{eq:strategy}). Note that by construction, we have that $\scheduler_1$ follows $\vect{x}$.
\end{example}

%% file: sections/stochastic_templates.tex
This section presents our novel \ASCerts. We carefully combine the certificates for stochastic invariants (\Cref{section:certificates}) and the PeSTels for almost-sure winning from \cite{phalakarn_winning_2024,phalakarn_templates_2025} to obtain \ASCerts in \Cref{sec:combine} and then devise a synthesis algorithm for them in \Cref{sec:synth1}. For the important class of \emph{globally optimal} strategies, we show that \ASCerts are \emph{relatively complete} and devise a simplified synthesis procedure for such optimal \ASCerts in \Cref{subsection:global-opt}. Lastly, \Cref{sec:extraction} gives a new formalization of randomized strategies which are extracted from \ASCerts.%

\subsection{Combining Certificates with Templates}\label{sec:combine}
Conceptually, we will combine certificates $\vect{x} \in \certs_{\geq \lambda}(\sg, I)$ with strategy templates by restricting the SG $(\sg,\Phi)$ to $I$ and collapsing all states outside of $I$ into a rejecting sink $\bot$. Synthesizing an almost-sure winning strategy template $\Lambda$ over the resulting subgame and combining it with $\vect{x}$ yields \ASCerts.

Towards this goal we first define the notion of \emph{subgames}. For a certificate $\vect{x} \in [0,1]^\states$ we define $\SA_{\vect{x}} \subseteq \SA$ to contain a state-action pair $(\state, \action)$ if and only if there exists a distribution $\mu \in \distr{\actions(\state)}$ such that $\mu(\action) > 0$ and
\begin{align}
\textstyle\vect{x}(\state) \geq \sum_{b \in \actions(\state)} \sum_{\state' \in \states} \mu(b) \cdot \transRel(\state, b, \state') \cdot \vect{x}(\state'). \label{eq:admissible-distr}
\end{align}
The set $\SA_{\vect{x}}$ contains the state-action pairs that can be played with positive probability when following $\vect{x}$. Due to \cref{eq:opponent}, all actions of $\pTwo$ are kept.
\begin{definition}[Subgame]
Let $\sg = (\states, \states_1, \states_2, \init, \actions, \transRel)$ be an SG, $I \subseteq \states$, and $\vect{x} \in [0, 1]^\states$. Let $\subgame{I}{\vect{x}}=(I \union \{\bot\}, (\statesOne \cap I) \union \{\bot\}, \statesTwo \cap I, \actions, \transRel')$ be the \emph{subgame} induced by $I$ and $\vect{x}$ where $\delta'(\bot, \action, \bot) = 1$ for all $\action \in \actions$, and
\[
\transRel'(\state, \action, \state') =
\begin{cases}
\transRel(\state, \action, \states \setminus I) & \text{if } \state' = \exit \\
\transRel(\state, \action, \state') & \text{else}
\end{cases}
\]
for all $(\state, \action) \in \SA_{\vect{x}} \intersection (I \times \actions)$ and $\state' \in I \union \{ \exit \}$.
\end{definition}
Intuitively, we restrict the game to states in $I$ and state-action pairs that can be played when following the certificate $\vect{x}$. The transitions leading to states $\states \setminus I$ are redirected to a fresh absorbing state $\exit$. With this definition in place, we are now ready to present our \ASCertsLong (\ASCerts).
\begin{definition}[\ASCertsLong (\ASCerts)]
Let $(\sg, \Phi)$ be a stochastic parity game and $\lambda \in [0, 1]$.
An \ASCert for $(\sg, \Phi)$ and $\lambda$ is a tuple $\stochasticTemplate = (I, \vect{x}, \template)$ where $I \subseteq \states$, $\vect{x} \in \certs_{\geq \lambda}(\sg, I)$, and $\template$ is an \emph{almost-sure winning strategy template} for $(\subgame{I}{\vect{x}}, \Phi \lor \exit)$ from every state in $\subgame{I}{\vect{x}}$. A $\pOne$ strategy $\scheduler_1$ \emph{follows} $\stochasticTemplate$ if it $(i)$ \emph{follows} $\vect{x}$, and $(ii)$ \emph{follows} $\template$ almost surely.
\label{def:stochastic-template-follow}
\end{definition}
\ASCertsLong combine the certificates for stochastic invariants with the strategy templates for almost-sure winning for parity objectives. Note that $\Phi \lor \eventually \exit$ can be directly encoded as a parity objective since $\exit$ is absorbing. The following result shows that our \ASCertLong can be used to capture quantitatively winning strategies.
\begin{restatable}[Soundness]{theorem}{stochasticTemplateSoundness}
Let $\stochasticTemplate$ be an \ASCert for $(\sg, \Phi)$ and $\lambda \in [0, 1]$. If $\scheduler_1$ is a strategy for $\pOne$ that follows $\stochasticTemplate$, then %
$\inf_{\scheduler_2 \in \schedulersTwo}\prob_{\sg, \init}^{\scheduler_1, \scheduler_2}(\Phi) \geq \lambda.$
\label{theorem:soundness}
\end{restatable}%
\noindent The certificate $\vect{x}$ proves that $\scheduler_1$ keeps the plays within $I$ with probability $\geq \lambda$. Since $\scheduler_1$ follows $\template$, we know that $\scheduler_1$ satisfies $\Phi$ or leaves $I$ almost surely. We combine these two insights to conclude that $\scheduler_1$ satisfies $\Phi$ with probability $\geq \lambda$.
\begin{remark}[Upper Bound]
Since parity objectives are closed under complementation, we can also capture strategies $\scheduler_1$ that satisfy
$\sup_{\scheduler_2 \in \schedulers_2} \prob_{\sg, \init}^{\scheduler_1, \scheduler_2}(\Phi) \leq \lambda$ via an \ASCert for $(\sg, \neg \Phi)$ and $1 - \lambda$.
\end{remark}
\subsection{\ASCert Synthesis}\label{sec:synth1}
So far we assumed the invariant $I$ to be given. We now address the construction of such an appropriate invariant. Conceptually, we follow the workflow presented in \cref{fig:ascert-synthesis-overview} and extract an invariant from a computed winning strategy as follows.
\begin{restatable}{lemma}{prefixInvariant} \label{lemma:prefix-indep-invariants}
Let $(\sg, \Phi)$ be an SG with a parity objective $\Phi$ and $\scheduler_1$ a memoryless strategy for Player 1. We define $I \subseteq \states$ by
\begin{align}
I = \{\state \in \states \mid \inf_{\scheduler_2 \in \schedulersTwo} \prob_{\sg, \state}^{\scheduler_1, \scheduler_2}(\Phi) > 0\}.\label{eq:invariant-def}
\end{align}
Then for all $\state \in \states$ we have
\begin{align}
&\inf_{\scheduler_2 \in \schedulersTwo}\prob_{\sg, \state}^{\scheduler_1, \scheduler_2}(\globally I) = \inf_{\scheduler_2 \in \schedulersTwo}\prob_{\sg, \state}^{\scheduler_1, \scheduler_2}(\Phi), \label{eq:inv-prob}\\
&\inf_{\scheduler_2 \in \schedulersTwo}\prob_{\sg, \state}^{\scheduler_1, \scheduler_2}\bigl( \Phi \lor \eventually (\states \setminus I) \bigr) = 1.
\end{align}
\label{lemma:prefix-invariant}
\end{restatable}
\noindent We note that a similar result for Markov chains has been presented in \cite{abate_quantitative_2025}. \Cref{lemma:prefix-indep-invariants} shows that, for a given memoryless strategy $\scheduler_1$ for $\pOne$, the states from which $\scheduler_1$ can win with positive probability correspond to the sought-after invariant. Specifically, the probability of remaining in $I$ corresponds to the probability of satisfying $\Phi$. Further, $\Phi$ is satisfied or $I$ is left almost surely. Recall that $\almostParityTemplate$ denotes the algorithm from \cite{phalakarn_templates_2025} which computes an almost-sure winning
strategy template. In total, \cref{lemma:prefix-indep-invariants} gives rise to the following approach for synthesizing an \ASCert:
\begin{enumerate}
\item Choose a memoryless $\pOne$ strategy $\scheduler_1$ such that $\inf_{\scheduler_2 \in \schedulersTwo} \prob_{\sg, \init}^{\scheduler_1, \scheduler_2}(\Phi) \geq \lambda$ and compute $I = \{\state \in \states \mid \inf_{\scheduler_2 \in \schedulersTwo} \prob_{\sg, \state}^{\scheduler_1, \scheduler_2}(\Phi) > 0\}$.
\item By \Cref{lemma:prefix-indep-invariants} we have that $\inf_{\scheduler_2 \in \schedulersTwo} \prob_{\sg, \init}^{\scheduler_1, \scheduler_2}(\globally I) \geq \lambda$. We can then synthesize a permissive certificate $\vect{x} \in \certs_{\geq \lambda}(\sg, I)$ (see \Cref{lemma:complete-mr} and \Cref{subsection:permissive-certificates}).
\item Further, from \Cref{lemma:prefix-indep-invariants} we know that $\inf_{\scheduler_2 \in \schedulers_2} \prob_{\subgame{I}{\vect{x}}, \init}^{\scheduler_1, \scheduler_2}(\Phi \lor \eventually \bot) = 1$. Thus, we can compute $\template = \almostParityTemplate(\subgame{I}{\vect{x}}, \Phi \lor \eventually \exit)$.
\item We return the \ASCert $T = (I, \vect{x}, \Lambda)$.
\end{enumerate}
We emphasize that the memoryless strategy $\scheduler_1$ is only needed as a starting point for the synthesis. The returned \ASCert $T$ captures \emph{all classes} of strategies that follow it, that is, $T$ is not restricted to the class of memoryless strategies, as we will further discuss in Sec.~\ref{sec:extraction}. The proposed synthesis approach can also be seen as a procedure for \emph{generalizing} single strategies to (possibly infinite) sets of strategies satisfying the objective, thereby enabling greater flexibility.

\changes{Let us now comment on the complexity of our synthesis procedure. In the first step, we need to determine an initial quantitatively winning strategy. The associated decision problem is in $\text{NP}{}\cap{}\text{coNP}$ \cite{chatterjee_quantitative_games_2004}. On the algorithmic side, e.g., \cite{ChatterjeeH06} solves stochastic parity games in randomized subexponential time. The second step can be performed in polynomial time via linear programming (see \Cref{subsection:permissive-lp}). Finally, in the third step, \almostParityTemplate{} first transforms the stochastic parity game into a non-stochastic one via the gadgets from \cite{chatterjee_quantitative_games_2004}, incurring a linear blowup in the number of priorities per state-action pair. The template is then computed on this game with the same complexity as Zielonka's algorithm \cite{anand_permissive_templates_2023}. Note that for reachability and (co-)Büchi objectives the template computation is more efficient, as the reduction does not depend on the number of priorities and the template can be computed in polynomial time \cite{anand_permissive_templates_2023}.}
\subsection{Globally Optimal \ASCerts}\label{subsection:global-opt}
We now consider the special case of using a globally optimal strategy as a starting point for \ASCerts synthesis from \Cref{sec:synth1}. It turns out that for \emph{globally optimal} strategies the choice of the invariant is \emph{unique}, i.e., for any strategy $\scheduler_1^*$ for $\pOne$ that is \emph{globally optimal} for a stochastic parity game $(\sg, \Phi)$ we have
\begin{align}
I^* = \{\state \in \states \mid \inf_{\scheduler_2 \in \schedulers_2} \prob_{\sg, \state}^{\scheduler_1^*, \scheduler_2}(\Phi) > 0\} = \{\state \in \states \mid \valueVector_{\sg,\Phi}(\state) > 0\}.
\label{eq:optimal-invariant}
\end{align}
For all globally optimal strategies, the invariants coincide and are inclusion-maximal. By \cref{eq:inv-prob} in \Cref{lemma:prefix-indep-invariants}, the certificate can be chosen as follows:
\begin{align}
\vect{x}^*(\state) = 1 - \valueVector_{\sg, \Phi}(\state) & \qquad \text{for all } \state \in \states.
\label{eq:optimal-cert}
\end{align}%
We now construct an \ASCert that only captures globally optimal strategies. Conversely, there always exists a globally optimal strategy that is captured by such an \ASCert. We refer to this as a \emph{relatively complete} characterization of globally optimal strategies. Below we omit the threshold $\lambda$ as it is not relevant.
\begin{restatable}[Relative Completeness]{theorem}{relativeCompleteness}\label{theorem:relative-complete}
Let $(\sg, \Phi)$ be a stochastic parity game, $I^*$ and $\vect{x}^*$ defined as in \eqref{eq:optimal-invariant} and \eqref{eq:optimal-cert}. Let $\stochasticTemplate^* = (I^*, \vect{x}^*, \Lambda)$ be an \ASCert.%
\begin{itemize}
\item If a $\pOne$ strategy $\scheduler_1^*$ follows $T^*$, then it is \emph{globally} optimal for $(\sg, \Phi)$.
\item Every memoryless \emph{globally} optimal $\pOne$ strategy for $(\sg, \Phi)$ follows $\vect{x}^*$. Further, there exists a globally optimal strategy for $\pOne$ that follows $T^*$.
\end{itemize}
\end{restatable}%
\noindent If we choose $I^*$ and $\vect{x}^*$ as in \eqref{eq:optimal-invariant} and \eqref{eq:optimal-cert}, respectively, the strategies that are captured by our \ASCerts are guaranteed to be globally optimal. We also call $\stochasticTemplate^*$ an \emph{optimal} \ASCert. Conversely, each memoryless globally optimal strategy follows the certificate $\vect{x}^*$ and there exists a globally optimal strategy that follows $\stochasticTemplate^*$. We can only prove relative completeness, in part because the templates from \cite{anand_permissive_templates_2023} are incomplete, i.e., a globally optimal strategy might not follow $\Lambda$. %

\subsection{Strategy Extraction from \ASCerts}\label{sec:extraction}
Previous work on strategy templates focused on the extraction of memoryless strategies \cite{anand_permissive_templates_2023} or a simple, counter-based strategy extraction \cite{phalakarn_winning_2024}. Given the context of quantitative winning, we now consider the extraction of general randomized strategies, by providing a more instructive characterization of what it means for strategies to follow an \ASCert. As an important by-product, we obtain a new characterization for strategies that follow a template almost surely.
For a $\pOne$ strategy $\scheduler_1$, finite play $\SApref \state \in \SA^* \states_1$, and set $A \subseteq \SA$, we write $\scheduler_1(\SApref \state)(A)$ for $\sum_{\action \in \actions(\state)} \scheduler_1(\SApref \state)(\action) \cdot \vect{1}_A(\state, \action)$, i.e., the probability of $\scheduler_1$ choosing state-action pairs from $A$ given $\SApref \state$. With this, we say that a (randomized) $\pOne$ strategy $\scheduler_1$ \emph{complies} with an \ASCert~$(I, \vect{x}, \template)$, where $\template=(U, D, \mathcal{H})$, if

\begin{subequations}\label{eq:strat_compliance}
 \begin{align}
&\vect{x}(\state) \geq \sum_{\action \in \actions(\state)} \sum_{\state' \in \states} \scheduler_1(\SApref \state)(\action) \cdot \transRel(\state, \action, \state') \cdot \vect{x}(\state') & \forall \SApref \state \in \playsFin^{\scheduler_1} \intersection \SA^{*} \states_1, \label{eq:cert-follow} \\
& \scheduler_1(\play)(U) = 0 & \forall \play \in \playsFin^{\scheduler_1} \intersection \SA^{*}\states_1, \label{eq:unsafe-sum} \\
& \sum_{\{n \mid \play_{=n} \in \src{D}\}} \scheduler_1(\play_{\leq n})(D) < \infty & \forall \play \in \plays^{\scheduler_1}, \label{eq:colive-sum}\\
& \play {\models} \globally \eventually \src{H} \Rightarrow \sum_{\{n \mid \play_{{=}n} \in \src{H} \}} \scheduler_1(\play_{\leq n})(H) = \infty& \forall H \in \mathcal{H},\; \forall \play \in \plays^{\scheduler_1}. \label{eq:live-sum}
\end{align}
\end{subequations}
Recall that \cref{eq:cert-follow} is the defining condition for following a certificate in \cref{def:follow-cert}. The conditions \eqref{eq:unsafe-sum}--\eqref{eq:live-sum} provide the novel characterization for following a strategy template $\template$ almost surely. Intuitively, \eqref{eq:unsafe-sum} requires that $\scheduler_1$ never plays any unsafe state-action pair. Further, the probability of playing a co-live pair must converge to a finite value (\cref{eq:colive-sum}), reflecting that such pairs are only chosen finitely often. Whenever a live-group $H \in \mathcal{H}$ is enabled infinitely often in a play $\play$, the total probability of choosing a pair from $H$ has to go to infinity (\cref{eq:live-sum}). This ensures that the probability of choosing $H$ does not vanish over time. The following result shows soundness of the characterization.
\begin{restatable}{theorem}{almostTemplateSoundness}\label{theorem:correctness-compliant-strat}
Let $(\sg, \Phi)$ be a stochastic parity game and $\lambda \in [0, 1]$. Further, let $\stochasticTemplate = (I, \vect{x}, \template)$ be an \ASCert for $(\sg, \Phi)$ and $\lambda$. Let $\scheduler_1$ be a $\pOne$ strategy that is compliant with $\stochasticTemplate$. Then, $\scheduler_1$ follows $\stochasticTemplate$ and we have $\inf_{\scheduler_2 \in \schedulers_2} \prob_{\sg, \init}^{\scheduler_1, \scheduler_2}(\Phi) \geq \lambda$.
\end{restatable}%
\noindent The characterization in \eqref{eq:strat_compliance} provides us with a recipe for \emph{extracting strategies}. When a state $\state$ is entered, we collect the actions $A \subseteq \actions(\state)$ that can always be played (according to $\template$), i.e., $\action \in A$ iff $(\state, \action) \notin U \union D$. Then, we consider the actions $A_{\mathcal{H}} \subseteq \actions(\state)$ that belong to a live-group and simply choose a distribution $\mu$ over $A \union A_{\mathcal{H}}$ such that the condition \eqref{eq:cert-follow} is satisfied and $\mu(A_{\mathcal{H}}) > 0$ if $A_{\mathcal{H}} \neq \emptyset$. To ensure that the probability of playing live-actions does not vanish over time, we determine a lower bound $p \in (0, 1]$ a priori and require the probability of playing such actions to be at least $p$. Note that the choice of the distribution $\mu$ is fully local and \emph{does not} depend on choices in other states.
For the special case of memoryless strategies, the quantification in \cref{eq:cert-follow} and \cref{eq:unsafe-sum} simplifies to a quantification over $\states_1$ and allows us to choose a fixed compliant distribution for each $\state\in\states_1$. The following result states that for \ASCerts that are synthesized as in \Cref{sec:synth1}, it is always possible to extract a compliant strategy.
\begin{restatable}{lemma}{conflictFree}\label{lemma:conflict-free}
	Let $\stochasticTemplate$ be an \ASCert synthesized as described in \Cref{sec:synth1}. Then, there exists a strategy $\scheduler_1$ which is compliant with $\stochasticTemplate$.
\end{restatable}

\begin{example}
We revisit our example from \Cref{section:overview} (see \cref{fig:running-example}) with objective $\Phi$ and threshold $\lambda = 0.75$. Combining the invariant $I = \states \setminus \{\state_4\}$, certificate $\vect{x}$ from \cref{eq:example-certificate} and template $\template = (U, D, \mathcal{H}) = \almostParityTemplate(\subgame{I}{\vect{x}}, \Phi \lor \eventually \exit)$, where $U = \emptyset$, $D = \{ \texttt{visit} \}$, and $\mathcal{H} = \{\{\texttt{maintain}\}\}$, yields the $\ASCert$ $\stochasticTemplate = (I, \vect{x}, \template)$. We can then use $\stochasticTemplate$ to easily extract strategic choices. For example, in $\state_5$ the certificate $\vect{x}$ enables us to pick any action. The template $\template$ informs us that we should play \texttt{maintain} with positive probability and avoid choosing \texttt{visit}. %
\end{example}

%% file: sections/evaluation.tex
We provide a \emph{proof-of-concept} Python tool that implements the synthesis of \ASCerts (\cref{fig:ascert-synthesis-overview} and \cref{sec:synth1}) and accepts SGs specified in the \textsc{Prism} language \cite{KNP11}.
Recall that the synthesis of \ASCerts for a stochastic parity game $(\sg, \Phi)$ requires an initial strategy. %
For the special case of reachability, we use \textsc{Prism-games} (via the \textsc{UMB} \cite{andriushchenko2026umbunifiedmarkovbinary} format) to compute globally optimal strategies. For general parity objectives, the tool landscape is scarce, and we use \textsc{Mungojerrie} \cite{hahn_mungojerrie_2023} to obtain a strategy via reinforcement learning\footnote{\textsc{Mungojerrie} contains a model checker for stochastic parity games. Unfortunately, we encountered instances where it was inconsistent with the learning algorithm.}. Since this does not yield optimal strategies in general, the \ASCerts synthesized by our tool are only \emph{globally optimal} (\Cref{subsection:global-opt}) for reachability objectives. Once a strategy $\scheduler_1$ is obtained, the invariant~$I$ (\cref{eq:invariant-def}) is computed by solving the induced MDP $\sg^{\scheduler_1}$ via \textsc{Storm} \cite{HenselJKQV22}. The synthesis of a permissive certificate $\vect{x}$ (\Cref{subsection:permissive-certificates}) is implemented using \textsc{Gurobi} \cite{gurobi}. 
Afterwards, the subgame $\sg[I,\vect{x}]$ is constructed and an almost-sure winning strategy template $\Lambda$ is computed over this subgame w.r.t.\ the objective $ \Phi \lor \eventually \exit$ by using the gadgets from \cite{chatterjee_quantitative_games_2004} and the \textsc{PeSTel} tool \cite{anand_permissive_templates_2023}. Finally, combining $\Lambda$, $I$ and $\vect{x}$ as in \cref{def:stochastic-template-follow} results in the desired \ASCert.

All experiments were performed on a machine with an Apple M2 CPU and 24GB memory. Our prototype and the experimental data are available at \cite{baier_2026_22012734}. %

\smallskip

\noindent\textbf{Performance.} We evaluated the performance of our prototype on six case studies (with between $257$ and $34645$ states) from the literature \cite{10.1007/978-3-030-53291-8_21,KwiatkowskaNP12,kwiatkowska_prism_games_2020,KretinskyRSW22,BalsEKW24,ChenFKPS13}. In \Cref{tab:parity-small} we summarize the results for parity objectives, while the results for reachability objectives are reported in \Cref{subsection:performance}. Each row in \Cref{tab:parity-small} corresponds to an SG $\sg$ with parity objective $\Phi$ (the number of colors in $\Phi$ is given by \emph{\#Colors}). The columns \emph{Invariant}, \emph{Certificate} and \emph{Template} show the runtime of each corresponding computation step and \emph{Total} describes the total \ASCert synthesis time.
For the considered models, the synthesis times for our \ASCerts ranged from a few seconds to roughly a minute\footnote{This excludes the time for synthesizing an initial strategy, since the runtime of \textsc{Mungojerrie} depends on the choice of (manually-tuned) hyperparameters.}.
Computing the certificate and strategy template are the most time-consuming steps, since they involve the setup of an LP and the construction of a non-stochastic parity game, respectively. Since both steps are currently implemented in Python, there is considerable room for improvement. We also observe that the number of colors affects the computation time of the template. The other steps seem mostly unaffected. In summary, our prototype can synthesize \ASCerts for (moderately-sized) models in a reasonable amount of time. We refer to \Cref{subsection:performance} for more details.
\input{tables/parity-small.tex}

\smallskip

\changes{\noindent\textbf{Runtime Adaptation.} We now consider the use of \ASCerts for adapting to unforeseen conditions at \emph{runtime} on a modified version of the \emph{hallway} model from \cite{10.1007/978-3-030-53291-8_21}. The model represents the interaction between a \emph{robot} ($\pOne$) and a \emph{human} ($\pTwo$) on a grid, as shown in \cref{fig:grid}. The robot can move north, south, west, or east, where for each move there is a non-zero probability that the robot instead ends up in a different direction. Similarly, the human can move on the grid. The robot's objective is to reach the gray area and maintain (i.e., visit) the machine~{\machine} infinitely often if it is used infinitely often by the human. This objective can be expressed as a parity objective $\Phi$. To reach the gray area, the robot must avoid the puddles~\puddle, or it will get stuck. Further, after maintenance, the robot needs to recharge at the station~{\station} before it can maintain the machine again. For simplicity, we assume that the human only moves when the robot is in the gray area, and that the gray area, once reached, cannot be left.

Suppose at \emph{runtime} the robot is equipped with sensors that measure the \emph{mechanical wear} on its components resulting from choosing the different actions. For simplicity, we assume that at each step and for each action a measurement $w$ is sampled uniformly from $\{n \in \reals \mid 1 \leq n \leq 10\}$, where higher values of $w$ correspond to higher wear. Using our \ASCerts, the robot can minimize the mechanical wear while still completing its objective $\Phi$ with sufficiently high probability.
To this end, when arriving in a grid location, the robot extracts a distribution $\mu$ over the available actions from the \ASCert that minimizes wear. This minimization respects the \ASCert and ensures that $\Phi$ will still be satisfied with the desired probability. We emphasize that the extraction and minimization are \emph{fully local}, i.e., they do not depend on decisions in other states.

We first obtained a strategy $\scheduler_1$ for which $\inf_{\scheduler_2 \in \schedulersTwo}\prob_{\sg}^{\scheduler_1, \scheduler_2}(\Phi) = \lambda \approx 0.86$ holds. Using $\scheduler_1$ as the initial strategy, we computed \ASCerts for different probability thresholds. That is, for thresholds $\lambda_{\gamma} = \lambda \cdot \gamma$ where $\gamma \in \{1, 0.75, 0.5, 0.25, 0\}$, we computed \ASCerts $\stochasticTemplate_\gamma$ where the certificate maximizes permissiveness in grid locations around the first puddle {\puddle} (as described in \Cref{subsection:permissive-certificates} and \Cref{subsection:permissive-lp}). Intuitively, in these locations in particular, the robot can perform riskier moves (as in the example in \Cref{section:overview}), so we provide the robot with more flexibility if higher failure probabilities can be tolerated. We then compared our approach against the (non-adaptive) initial strategy $\scheduler_1$. For both approaches, we ran $2000$ simulation runs, where each run consisted of $1000$ steps. The actions of the human ($\pTwo$) were chosen uniformly at random. The results in \cref{fig:runtime-adaptation} show that the \ASCerts enable the robot to reduce the total wear substantially. Notably, even for $\gamma = 1$, the wear can already be reduced by $40\%$ compared to the fixed strategy $\scheduler_1$. Note that $\stochasticTemplate_{1}$ is an \ASCert for all thresholds $\lambda_\gamma$. However, when the probability threshold is relaxed, i.e., $\gamma<1$, our synthesis approach yields \ASCerts $\stochasticTemplate_\gamma$ which are even more permissive than $\stochasticTemplate_{1}$. In particular, as shown in \cref{fig:runtime-adaptation}, the more the probability threshold is relaxed, the greater the reduction in mechanical wear, indicating the usefulness of the permissive certificate synthesis. Altogether, this demonstrates the effectiveness and potential of \ASCerts for runtime adaptation. More details are provided in \Cref{subsection:adaptation}.
}

\begin{figure}[!t]
	\centering
	\hfill
	\begin{minipage}[b]{0.45\textwidth}
		\centering
		\includegraphics[height=3.55cm]{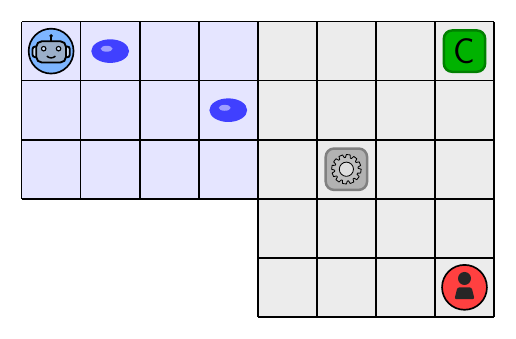}
		\vspace{0cm}
		\caption{Illustration of the grid.}
		\label{fig:grid}
	\end{minipage}
	\hfill
	\begin{minipage}[b]{0.49\textwidth}
		\centering
		\includegraphics[height=4cm]{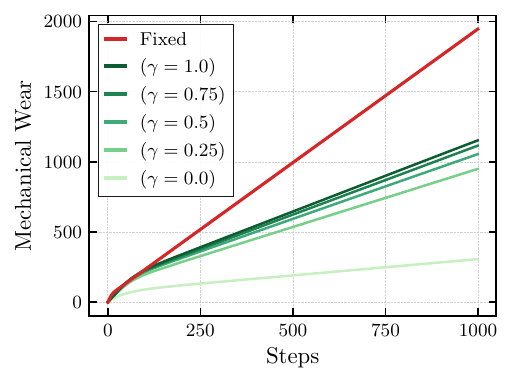}
		\vspace{-0.2cm}
		\caption{Average mechanical wear.}
		\label{fig:runtime-adaptation}
	\end{minipage}
\end{figure}

%% file: tables/parity-small.tex
{
\renewcommand{\arraystretch}{0.7} %
\begin{table}[t] \scriptsize
	\centering
	\setlength{\tabcolsep}{4pt}  %
	\caption{Runtimes for parity (wall-clock times and specified in seconds).}
	\label{tab:parity-small}
	\begin{tabular}{l @{\hspace{3em}} ccc @{\hspace{3em}} rrrrr}
		\toprule
		Case Study & \#Colors & $\card{\states}$ & $\card{\SA}$ & Invariant & Certificate & Template & Total \\
\midrule
\multirow{9}{*}{\texttt{bigmec} \cite{KretinskyRSW22}}& \multirow{3}{*}{3} & 10003 & 21432 & 1.133 & 11.378 & 1.613 & 14.303 \\
&  & 20003 & 42860 & 4.301 & 28.648 & 4.297 & 37.351 \\
&  & 30003 & 64290 & 9.740 & 49.046 & 8.155 & 67.077 \\
\cmidrule{2-8}
& \multirow{3}{*}{5} & 10003 & 21432 & 1.217 & 11.398 & 1.650 & 14.269 \\
&  & 20003 & 42860 & 4.566 & 28.390 & 4.467 & 37.430 \\
&  & 30003 & 64290 & 10.128 & 48.760 & 8.559 & 67.459 \\
\midrule
\multirow{3}{*}{\texttt{grid} \cite{BalsEKW24}}& \multirow{3}{*}{3} & 257 & 1029 & 0.011 & 0.577 & 0.129 & 0.949 \\
&  & 1025 & 4100 & 0.051 & 2.307 & 0.490 & 2.850 \\
&  & 4097 & 16388 & 0.340 & 9.228 & 2.193 & 11.767 \\
\midrule
\multirow{4}{*}{\texttt{hallway} \cite{10.1007/978-3-030-53291-8_21}}& \multirow{2}{*}{3} & 10240 & 26624 & 1.308 & 9.978 & 2.553 & 13.839 \\
&  & 25000 & 65000 & 6.937 & 24.857 & 9.291 & 41.136 \\
\cmidrule{2-8}
& \multirow{2}{*}{5} & 10240 & 26624 & 1.331 & 9.936 & 3.687 & 14.962 \\
&  & 25000 & 65000 & 7.166 & 24.941 & 13.363 & 45.488 \\
\end{tabular}
\end{table}
}

%% file: sections/conclusion.tex
\changes{This work introduced \ASCertsLong (\ASCerts) as a local and permissive representation of quantitatively winning strategies in stochastic parity games. 
	We first presented certificates for stochastic invariants in SGs and showed their usage as a permissive representation of strategies. Our \ASCerts then carefully combined these certificates with strategy templates \cite{phalakarn_templates_2025,phalakarn_winning_2024,anand_computing_2023} to capture quantitatively winning strategies.
	The proposed \ASCert synthesis starts from a winning strategy, synthesizes a permissive certificate and computes the template via an induced subgame. For globally optimal strategies the synthesis can be simplified and \ASCerts provide a relatively complete characterization for such strategies. Finally, we presented the extraction of strategies from \ASCerts and empirically demonstrated their usefulness for runtime adaptation.

We plan on using \ASCerts for \emph{shielding} \cite{AlshiekhBEKNT18,bloem_shield_2015} in reinforcement learning to enforce that learned strategies satisfy an $\omega$-regular objective with a guaranteed probability. While \cite{anand2025followstarsdynamicomegaregular} uses PeSTels for shielding, it can only provide qualitative guarantees. Lastly, our tool would benefit from engineering improvements and an integrated stochastic parity game solver.
}

%% file: credits.tex
\begin{credits}
	\changes{\subsubsection{\ackname}
	This work was funded by the German Research Foundation (DFG) under 
	Germany's Excellence Strategy: EXC 2050/2, 390696704 -- Cluster of Excellence ``Centre for Tactile Internet with Human-in-the-Loop'' (CeTI) of TU Dresden;
	by DFG grant 389792660 as part of TRR~248 (see \url{https://perspicuous-computing.science}); 
	by the German Federal Ministry of Research, Technology and Space (BMFTR) within the project SEMECO Q1 (03ZU1210AG); by the DFG project SCHM 3541/1-1.
}
\subsubsection{\discintname}
The authors have no competing interests to declare that are relevant to the content of this article.
\end{credits}

%% file: sections/appendix-certificates.tex
We now provide a soundness proof for our certificates (\Cref{theorem:cert-invariant}). Further, we present the LP for synthesizing permissive certificates (\Cref{subsection:permissive-lp}).

\stochasticInvariantCertificates*
\begin{proof}
Let $\scheduler_2$ be an arbitrarily chosen $\pTwo$ strategy.
Let $X = \{X_t\}_{t \in \naturals}$ be the Markov chain induced by $\scheduler_1$ and $\scheduler_2$. We write $\prob(\cdot)$ to denote the probability measure of the associated probability space. Further, $\expectation{}{\cdot}$ denotes the expectation. Now let $Y = \{Y_t\}_{t \in \naturals}$ be a stochastic process  where $Y_t \colon \plays(\sg^{\scheduler_1, \scheduler_2}, \init) \to \realsnn$ is defined by 
\[
Y_t(\play) = 
\begin{cases}
\vect{x}(\state_t) & \text{if } \forall i \in \{0, \ldots, t\} \centerdot \state_i \notin \states \setminus I \\
1 & \text{else}
\end{cases}
\]
for all $\play \in \plays(\sg^{\scheduler_1, \scheduler_2}, \init)$ and $t \in \naturals$. For all $t \in \naturals$ we have:
\begin{itemize}
\item $\expectation{}{\card{Y_t}} < \infty$ 
\item $\expectation{}{Y_{t+1} \mid X_0, \ldots, X_t} \leq Y_t$ (holds almost surely)
\end{itemize}
To see that the latter holds, let $t \in \naturals$ and $\play = \state_0 \action_0 \state_1 \ldots \in \plays(\sg^{\scheduler_1, \scheduler_2}, \init)$ be arbitrarily chosen. We then make the following case distinction.
\begin{description}
\item[Case 1: $\state_i \notin \states \setminus I$ for all $0 \leq i \leq t$.]
\begin{align*}
&\phantom{{}={}}\expectation{}{Y_{t+1} \mid X_0, \ldots, X_t}(\play) \\
&= 
\begin{cases}
\sum_{\action \in \actions(\state_t)} (\scheduler_1(\state_0 \ldots \state_t)(\action) \cdot \sum_{\state' \in \states} \transRel(\state_t, \action, \state') \cdot \vect{x}(\state')) & \text{if } \state_t \in \statesOne\\
\sum_{\action \in \actions(\state_t)} (\scheduler_2(\state_0 \ldots \state_t)(\action) \cdot \sum_{\state' \in \states} \transRel(\state_t, \action, \state') \cdot \vect{x}(\state')) & \text{if } \state_t \in \statesTwo
\end{cases} \\
&\leq \vect{x}(\state_t) = Y_t(\play)
\end{align*}
\item[Case 2: $\state_i \in \states \setminus I$ for some $0 \leq i \leq t$.]
\begin{align*}
\phantom{{}={}}\expectation{}{Y_{t+1} \mid X_0, \ldots, X_t}(\play)
&= 1 = Y_t(\play)
\end{align*}
\end{description}
Thus, $Y$ is a \emph{supermartingale} w.r.t. $X$. By \cite[Theorem 5]{zikelic_rasm_2023} we have that
\begin{align*}
 \expectation{}{Y_0} \geq \prob\left(\sup_{t \geq 0} Y_t \geq 1\right) \geq \prob(\eventually (\states \setminus I)).
\end{align*}
By construction we have $\expectation{}{Y_0} = \vect{x}(\init) \leq 1 - \lambda$.
We can thus conclude that $\prob_{\sg, \init}^{\scheduler_1, \scheduler_2}(\globally I) \geq \lambda$. Since the choice of $\scheduler_2$ was arbitrary, we can conclude that $\inf_{\scheduler_2'} \prob_{\sg, \init}^{\scheduler_1, \scheduler_2'}(\globally I) \geq \lambda$. \qed
\end{proof}
\subsection{Linear Program for Synthesizing Permissive Certificates}\label{subsection:permissive-lp}
The linear program below is a heuristic approach for synthesizing permissive certificates as discussed in \Cref{subsection:permissive-certificates}.
\[
\begin{aligned}
	\max & \sum_{\state \in \states_1} \varepsilon_{\state} &&\\
	\text{s.t.}& \qquad \vect{x} \in [0,1]^{\states},\; \varepsilon_{\state} \in [0, 1]\; \forall \state \in \states_1 &&\\
	\vect{x}(\init) &\leq 1 - \lambda, && \\
	\vect{x}(\state) &\geq \sum_{\state' \in \states} \transRel(\state, \action, \state') \cdot \vect{x}(\state') && \forall \state \in I \intersection \statesTwo, \forall \action \in \actions(\state),\\
	\vect{x}(\state) &= 1 && \forall \state \in \states \setminus I, \\
	\vect{x}(\state) &\geq \sum_{\action \in \actions(\state)} \sum_{\state' \in \states} \scheduler_1(\state)(\action) {\cdot} \transRel(\state, \action, \state') {\cdot} \vect{x}(\state') && \forall \state \in I \intersection \states_1, \\
	\textstyle\vect{x}(\state) &\geq \left( \sum_{\state' \in \states} \transRel(\state, \action, \state') \cdot \vect{x}(\state') \right) - (1 - \varepsilon_{\state}) && \forall \state \in \states_1,\; \forall \action \in \actions(\state). &&
\end{aligned}
\]
A natural variation of the linear program is to instead consider the objective $\sum_{\state \in \states_1'} \varepsilon_{\state}$ where $\states_1' \subseteq \states_1$ is a designated set of states where permissiveness (or flexibility) is needed the most.

%% file: sections/appendix-stochastic-template.tex
We first provide a proof for the soundness of our \ASCerts (\Cref{theorem:soundness}) and for \Cref{lemma:prefix-invariant} that enables the synthesis of \ASCerts. Before we prove the relative completeness result, we prove the correctness of our characterization of compliant strategies (\Cref{theorem:correctness-compliant-strat}) and show that it is always possible to extract a strategy from our \ASCerts when following the synthesis workflow from \Cref{sec:synth1}. Finally, we prove the relative completeness of \ASCerts w.r.t. globally optimal strategies (\Cref{theorem:relative-complete}).

\subsection{Soundness of \ASCerts}
\stochasticTemplateSoundness*
\begin{proof}
	Since $\scheduler_1$ follows $\vect{x}$ we know that
	\[
	\inf_{\scheduler_2 \in \schedulersTwo}\prob_{\sg, \init}^{\scheduler_1, \scheduler_2}(\globally I) \geq \lambda.
	\]
	Further, we know that the subgame $\subgame{I}{\vect{x}}$ preserves all state-action pairs that can be chosen by a strategy following $\vect{x}$. Hence, we can apply $\scheduler_1$ to the subgame $\subgame{I}{\vect{x}}$ and have that
	\[
	\sup_{\scheduler_2 \in \schedulersTwo}\prob_{\subgame{I}{\vect{x}}, \init}^{\scheduler_1, \scheduler_2}(\eventually \exit) \leq 1 - \lambda.
	\]
	Since $\scheduler_1$ follows $\Lambda$ almost surely, it also does so in the subgame $\sg[I,\vect{x}]$. By \cite[Theorem 6]{phalakarn_winning_2024} we have
	\[
	\inf_{\scheduler_2 \in \schedulersTwo}\prob_{\subgame{I}{\vect{x}}, \init}^{\scheduler_1, \scheduler_2}(\Phi \lor \eventually \bot) = 1.
	\]
	Thus, we can conclude that
	\[
	\inf_{\scheduler_2 \in \schedulersTwo}\prob_{\subgame{I}{\vect{x}}, \init}^{\scheduler_1, \scheduler_2}(\Phi) \geq \lambda.
	\]
	Lastly, every play in $\subgame{I}{\vect{x}}$ that reaches $\exit$ cannot satisfy the parity objective $\Phi$. On the other hand, all plays remaining in $I$ are also plays in $\sg$. Hence we have that $\inf_{\scheduler_2 \in \schedulers_2} \prob_{\subgame{I}{\vect{x}}, \init}^{\scheduler_1, \scheduler_2}(\Phi) \leq \inf_{\scheduler_2 \in \schedulers_2} \prob_{\sg, \init}^{\scheduler_1, \scheduler_2}(\Phi)$.
	\qed
\end{proof}

\prefixInvariant*
\begin{proof}
	Let $\state \in \states$ be arbitrarily chosen. We show $\inf_{\scheduler_2 \in \schedulersTwo}\prob_{\sg, \state}^{\scheduler_1, \scheduler_2}\bigl( \Phi \lor \eventually (\states \setminus I) \bigr) = 1$ first. Let $\scheduler_2$ be an arbitrary Player 2 strategy. By law of total probability:
	\begin{align*}
		\prob_{\sg, \state}^{\scheduler_1, \scheduler_2}\bigl( \Phi \lor \eventually (\states \setminus I) \bigr) &= \prob_{\sg, \state}^{\scheduler_1, \scheduler_2}\bigl( \Phi \lor \eventually (\states \setminus I) \mid \neg \globally I \bigr) \cdot \prob_{\sg,\state}^{\scheduler_1, \scheduler_2}(\neg \globally I) {}+{} \\
		&{}\phantom{{}={}} \prob_{\sg, \state}^{\scheduler_1, \scheduler_2}\bigl( \Phi \lor \eventually (\states \setminus I) \mid \globally I \bigr) \cdot \prob_{\sg,\state}^{\scheduler_1, \scheduler_2}(\globally I)\\
		& = \prob_{\sg,\state}^{\scheduler_1, \scheduler_2}(\neg \globally I) + \prob_{\sg, \state}^{\scheduler_1, \scheduler_2}\bigl( \Phi \mid \globally I \bigr) \cdot \prob_{\sg,\state}^{\scheduler_1, \scheduler_2}(\globally I)
	\end{align*}
	Recall that $I = \{\state \in \states \mid \inf_{\scheduler_2 \in \schedulersTwo} \prob_{\sg, \state}^{\scheduler_1, \scheduler_2}(\Phi) > 0\}$. Thus, there exists $\gamma > 0$ s.t. each time a state $\state \in I$ is visited, there is a probability $\geq \gamma$ to satisfy $\Phi$. Hence the probability to never satisfy $\Phi$ when staying in $I$ converges to zero and we have:
	\[
	\prob_{\sg, \state}^{\scheduler_1, \scheduler_2}\bigl( \Phi \mid \globally I \bigr) = 1.
	\]
	In total this gives us $\prob_{\sg, \state}^{\scheduler_1, \scheduler_2}\bigl( \Phi \lor \eventually (\states \setminus I) \bigr) = \prob_{\sg,\state}^{\scheduler_1, \scheduler_2}(\neg \globally I) + \prob_{\sg,\state}^{\scheduler_1, \scheduler_2}(\globally I) = 1$.
	
	We now prove $\inf_{\scheduler_2 \in \schedulersTwo}\prob_{\sg, \state}^{\scheduler_1, \scheduler_2}(\globally I) = \inf_{\scheduler_2 \in \schedulersTwo}\prob_{\sg, \state}^{\scheduler_1, \scheduler_2}(\Phi)$.
	Since $\Phi$ is a parity objective, its complement $\neg \Phi$ is also a parity objective.
	By definition, for all $\state \in \states \setminus I$ we have $\inf_{\scheduler_2 \in \schedulersTwo} \prob_{\sg, \state}^{\scheduler_1, \scheduler_2}(\Phi) = 0$, i.e. $\sup_{\scheduler_2 \in \schedulersTwo} \prob_{\sg, \state}^{\scheduler_1, \scheduler_2}(\neg \Phi) = 1$. Hence, for all $\state \in \states$ we have:
	\begin{align*}
		\inf_{\scheduler_2 \in \schedulersTwo}\prob_{\sg, \state}^{\scheduler_1, \scheduler_2}(\globally I) 
		&= 1 - \sup_{\scheduler_2 \in \schedulersTwo} \prob_{\sg, \state}^{\scheduler_1, \scheduler_2}(\eventually (\states \setminus I)) \\
		& \geq 1 - \sup_{\scheduler_2 \in \schedulersTwo} \prob_{\sg, \state}^{\scheduler_1, \scheduler_2}(\neg \Phi) \\
		& = \inf_{\scheduler_2 \in \schedulersTwo} \prob_{\sg, \state}^{\scheduler_1, \scheduler_2}(\Phi)
	\end{align*}
	Equality follows from $\inf_{\scheduler_2 \in \schedulersTwo}\prob_{\sg, \state}^{\scheduler_1, \scheduler_2}\bigl( \Phi \lor \eventually (\states \setminus I) \bigr) = 1$, which is equivalent to:
	\[
	\sup_{\scheduler_2 \in \schedulersTwo}\prob_{\sg, \state}^{\scheduler_1, \scheduler_2}\bigl( \neg \Phi \land \globally I \bigr) = 0.
	\] \qed
\end{proof}

\subsection{Compliant Strategies}
\almostTemplateSoundness*
\begin{proof}
	By assumption $\scheduler_1$ satisfies conditions \eqref{eq:cert-follow}--\eqref{eq:live-sum}. It is clear that $\scheduler_1$ follows $\vect{x}$ and we only need to prove that $\inf_{\scheduler_2 \in \schedulers_2} \prob_{\sg, \state}^{\scheduler_1, \scheduler_2}(\Lambda) = 1$ for every state $\state \in \states$. Let $\state \in \states$ be arbitrarily chosen. According to \eqref{eq:unsafe-sum}, $\scheduler_1$ never plays an unsafe action and we have
	\begin{align}
		\inf_{\scheduler_2 \in \schedulers_2} \prob_{\sg, \state}^{\scheduler_1, \scheduler_2}(\safetyTemplate(U)) = 1. 
	\end{align}
	Now suppose there exists $\scheduler_2$ such that $\prob_{\sg, \state}^{\scheduler_1, \scheduler_2}(\liveTemplate(\mathcal{H})) < 1$. Hence there exists $H \in \mathcal{H}$ such that $\prob_{\sg, \state}^{\scheduler_1, \scheduler_2}(\globally \eventually \src{H} \land \eventually \globally \neg H) > 0$. Let us define the indicator variable $X_n \colon \SA^\omega \to \{0, 1\}$ to be $1$ iff $\play_n \in H$, for all $n \geq 0$. Then, for all $n \geq 0$ and plays $\play$ compliant with $(\scheduler_1, \scheduler_2)$ we have 
	\[
	\prob_{\sg, \state}^{\scheduler_1, \scheduler_2}(X_n = 1 \mid \play_{\leq n}) = \scheduler_1(\play_{\leq n})(H).
	\]
	By \eqref{eq:live-sum}, we have that for any compliant play $\play$ with $\play \models \globally \eventually \src{H}$ we have
	\[
	\sum_{n \geq 0} \prob_{\sg, \state}^{\scheduler_1, \scheduler_2}(X_n = 1 \mid \play_{\leq n}) = \infty
	\]
	and it follows that
	\[\prob_{\sg, \state}^{\scheduler_1, \scheduler_2}\left(\globally \eventually \src{H} \implies \sum_{n \geq 0} \prob_{\sg, \state}^{\scheduler_1, \scheduler_2}(X_n = 1 \mid \play_{\leq n}) = \infty\right) = 1.
	\]
	By the \emph{Second Borel-Cantelli} lemma (see, e.g., \cite[Theorem~4.3.4]{Durret2019}), we can conclude that
	\begin{align*}
	\prob_{\sg, \state}^{\scheduler_1, \scheduler_2}\left(\globally \eventually \src{H} \implies \sum_{n\geq0} X_n = \infty \right) = 1.
	\end{align*}
	However, this contradicts our assumption that $\prob_{\sg, \state}^{\scheduler_1, \scheduler_2}(\globally \eventually \src{H} \land \eventually \globally \neg H) > 0$. Thus such $H$ cannot exist and in total we have:
	\begin{align}
		\inf_{\scheduler_2 \in \schedulers_2} \prob_{\sg, \state}^{\scheduler_1, \scheduler_2}(\liveTemplate(\mathcal{H})) = 1. 
	\end{align}
	Lastly, by \eqref{eq:colive-sum}, for all plays $\play \in \plays^{\scheduler_1}$ we have
	\begin{align}
		\sum_{\{n \mid \play_{=n} \in \src{D}\}} \scheduler_1(\play_{\leq n})(D) < \infty. \label{eq:proof-colive}
	\end{align}
	Let $\scheduler_2 \in \schedulers_2$ be arbitrarily chosen. Further, let us define the indicator variable $Y_n \colon \SA^\omega \to \{0, 1\}$ to be $1$ iff $\play_n \in D$, for all $n \geq 0$. Then, for all plays $\play$ that are compliant with $(\scheduler_1, \scheduler_2)$ and $n \geq 0$ we have
	\begin{align}
		\prob_{\sg, \state}^{\scheduler_1, \scheduler_2}(Y_n = 1 \mid \play_{\leq n}) = \scheduler_1(\play_{\leq n})(D). \label{eq:proof-colive-strat}
	\end{align}
	Combining \eqref{eq:proof-colive} and \eqref{eq:proof-colive-strat}, we have that for all plays $\play$ compliant with $(\scheduler_1, \scheduler_2)$:
	\[
	\sum_{n \geq 0} \prob_{\sg,\state}^{\scheduler_1,\scheduler_2}(Y_n = 1 \mid \play_{\leq n})< \infty.
	\]
	In particular, we thus have:
	\[
	\prob_{\sg, \state}^{\scheduler_1, \scheduler_2} \left( \sum_{n \geq 0} \prob_{\sg,\state}^{\scheduler_1,\scheduler_2}(Y_n = 1 \mid \play_{\leq n}) < \infty \right) = 1.
	\]
	Using the \emph{Second Borel-Cantelli} lemma \cite[Theorem~4.3.4]{Durret2019} we can conclude that
	\[
	1 = \prob_{\sg, \state}^{\scheduler_1, \scheduler_2} \left( \sum_{n \geq 0} Y_n < \infty \right) = \prob_{\sg, \state}^{\scheduler_1, \scheduler_2}(\coliveTemplate(D)).
	\]
	Since the choice of $\scheduler_2$ was arbitrary, we can conclude that:
	\[
	\inf_{\scheduler_2 \in \schedulers_2} \prob_{\sg, \state}^{\scheduler_1, \scheduler_2}(\coliveTemplate(D)) = 1. 
	\]
	Since $\state$ was arbitrary chosen, we can conclude that $\inf_{\scheduler_2 \in \schedulers_2} \prob_{\sg, \state}^{\scheduler_1, \scheduler_2}(\Lambda) = 1$ for all $\state \in \states$.
	Hence, $\scheduler_1$ follows $T$ and by \Cref{theorem:soundness} we get $\inf_{\scheduler_2 \in \schedulersTwo} \prob_{\sg, \init}^{\scheduler_1, \scheduler_2}(\Phi) \geq \lambda$.\qed
\end{proof}

\begin{lemma}\label{lemma:colive-state-action-pairs-cert}
	Let $\vect{x} \in \certs_{\geq \lambda}(\sg, I)$ be a certificate. Let $t \in \states_1$ and $b \in \actions(t)$ such that
	\[
	\vect{x}(t) > \sum_{\state' \in \states} \transRel(t, b, \state') \cdot \vect{x}(\state').
	\]
	If $\scheduler_1$ is a $\pOne$ strategy that follows $\vect{x}$, then $\sup_{\scheduler_2 \in \schedulersTwo}\prob_{\sg, \state}^{\scheduler_1, \scheduler_2}\left(\globally \eventually (t, b) \right) = 0$ for every $\state \in \states$.
\end{lemma}
\begin{proof}
	Suppose $\scheduler_1'$ follows $\vect{x}$ and, for the sake of contradiction, assume that $\sup_{\scheduler_2 \in \schedulersTwo} \prob_{\sg, \state}^{\scheduler_1', \scheduler_2}(\globally \eventually (t,b)) > 0$ for some $\state \in \states$. Then, there also exists a memoryless strategy $\scheduler_1$ that follows $\vect{x}$ such that $\sup_{\scheduler_2 \in \schedulersTwo} \prob_{\sg, \state}^{\scheduler_1, \scheduler_2}(\globally \eventually (t,b)) > 0$. Let us consider the induced MDP $\sg^{\scheduler_1} = (\states, \init, \actions', \transMat)$ where $\actions' = \actions \union \{\tau\}$ and
	\begin{align*}
		\transMat(\state, \action, \state') =
		\begin{cases}
			\transRel(\state, \action, \state') & \text{if } \state \in \states_2, \action \in \actions(\state)\\
			\sum_{\action \in \actions(\state)} \scheduler_1(\state)(\action) \cdot \transRel(\state, \action, \state') & \text{if } \state \in \states_1, \action = \tau, \\
			0 & \text{else.}
		\end{cases}
	\end{align*}
	By construction of $\sg^{\scheduler_1}$, we have for all $\state \in \states$ and $\action \in \actions'(\state)$:
	\begin{align}
		\vect{x}(\state) \geq \sum_{\state' \in \states} \transMat(\state, \action, \state') \cdot \vect{x}(\state'). \label{eq:strict-proof-1}
	\end{align}
	Since, by assumption, $(t, b)$ is chosen infinitely often, it has to be part of an end component \cite{Alfaro97} (also see, e.g., \cite[Def. 10.117]{baier_principles_2008}) of $\sg^{\scheduler_1}$, that is, a set of states $T$ and function $A \colon T \to 2^{\actions'}$ such that the induced sub-MDP is strongly connected. Then, for all $\state \in T$ and $\action \in A(s)$, we have
	\begin{align}
		\vect{x}(\state) \geq \sum_{\state' \in \states} \transMat(\state, \action, \state') \cdot \vect{x}(\state') \overset{(\dag)}{=} \sum_{\state' \in T} \transMat(\state, \action, \state') \cdot \vect{x}(\state'). \label{eq:strict-proof-5}
	\end{align}
	Above, $(\dag)$ follows from the fact that $(T, A)$ is closed under $\transMat(\state, \action, \cdot)$. Recall that by assumption on $t$ and $b$, we have
	\[
	\vect{x}(t) > \sum_{\state' \in T} \transRel(t, b, \state' ) \cdot \vect{x}(\state') \geq \min_{\state' \in T} \vect{x}(\state').
	\]
	Let us define $\gamma = \min_{\state' \in T} \vect{x}(\state')$. Then, for all states $\state \in T$ with $\vect{x}(\state) = \gamma$ there does not exist an action $\action \in A(\state)$, and $\state' \in T$ such that $\transMat(\state, \action, \state') > 0$ and $\vect{x}(\state') > \gamma$, as otherwise \cref{eq:strict-proof-5} would be violated. However, this contradicts the fact that $(T,A)$ is an EC and strongly connected. Thus $(t, b)$ is not part of an EC in $\sg^{\scheduler_1}$ and the probability of choosing $(t, b)$ infinitely often in $\sg^{\scheduler_1}$ is zero (see, e.g., \cite[Theorem 10.120]{baier_principles_2008}). Altogether, this implies that $\prob_{\sg, \state}^{\scheduler_1, \scheduler_2}(\globally \eventually (t, b)) = 0$ for all $\pTwo$ strategies $\scheduler_2$, contradicting the choice of $\scheduler_1$.
	\qed
\end{proof}
\conflictFree*
\begin{proof}
	By \Cref{lemma:prefix-indep-invariants}, all states in the subgame $\subgame{I}{\vect{x}}$ are almost-sure winning for the objective $\Phi \lor \eventually \bot$. We recall that $\almostParityTemplate(\subgame{I}{\vect{x}}, \Phi \lor \eventually \bot)$ computes an almost surely winning strategy template by reducing the stochastic parity game $(\subgame{I}{\vect{x}}, \Phi \lor \eventually \bot)$ to a non-stochastic parity game via the gadgets from \cite{chatterjee_quantitative_games_2004} and then uses the techniques from \cite{anand_permissive_templates_2023}. Let $\template \coloneqq \safetyTemplate(U) \land \coliveTemplate(D) \land \liveTemplate(\mathcal{H})$ where
	\begin{align}
		\safetyTemplate(U) &\textstyle\coloneqq \globally \bigwedge_{(\state, \action) \in U} \neg (\state, \action),\\
		\coliveTemplate(D) &\textstyle\coloneqq \bigwedge_{(\state, \action) \in D} \eventually \globally \neg (\state, \action),\text{ and}\\
		\liveTemplate(\mathcal{H}) &\textstyle\coloneqq \Conj_{H \in \mathcal{H}} \globally \eventually \src{H} \Rightarrow \globally \eventually H.
	\end{align}
	Since all states in the subgame are almost sure winning, we have that $U = \emptyset$. W.l.o.g., we may assume that $D$ only contains state-action pairs that can actually be chosen infinitely often by a strategy $\scheduler_1$ that follows $\vect{x}$. Specifically, for all $\state \in \states_1$ and $\action \in \actions(\state)$, we may assume that if
	\[
	\vect{x}(\state) > \sum_{\state' \in \states} \transRel(\state, \action, \state') \cdot \vect{x}(\state')
	\]
	then $(\state, \action) \notin D$, as by \Cref{lemma:colive-state-action-pairs-cert} for such state-action pairs we have
	\[\prob_{\sg, \init}^{\scheduler_1, \scheduler_2}(\globally \eventually (\state, \action)) = 0
	\]
	for any strategy profile $(\scheduler_1, \scheduler_2)$ where $\scheduler_1$ follows $\vect{x}$.
	
	We can then construct a memoryless strategy $\scheduler_1$ as described in \Cref{sec:extraction}. Observe that in each state $\state$ it is possible to choose an action distribution $\mu$, since by construction there exists a strategy for $\pOne$ that follows $\vect{x}$. Further, by \cite[Proposition 1]{anand_permissive_templates_2023} there always exists an action in $\state$ that is not co-live (or unsafe), i.e., the set of allowed (and live) actions $A \union A_\mathcal{H} \subseteq \actions(\state)$ is non-empty. Then for each state $\state \in \states_1$, we choose $\mu$ to be a distribution with support in $A \union A_\mathcal{H}$ that satisfies $\mu(A_\mathcal{H}) > 0$ if $A_\mathcal{H} \neq \emptyset$. If we have
	\[
	\vect{x}(\state) \geq \sum_{\action \in \actions(\state)} \sum_{\state' \in \states} \mu(\action) \cdot \transRel(\state, \action, \state') \cdot \vect{x}(\state')
	\]
	then we are done and set $\scheduler_1(\state) = \mu$. Suppose we have
	\[
	\vect{x}(\state) < \sum_{\action \in \actions(\state)} \sum_{\state' \in \states} \mu(\action) \cdot \transRel(\state, \action, \state') \cdot \vect{x}(\state')
	\]
	By construction of the subgame, every action included in the subgame can be played with positive probability (cf. \cref{eq:admissible-distr}). Thus, there exists $\action^* \in \actions(\state)$ such that
	\[
	\vect{x}(\state) > \sum_{\state' \in \states} \delta(\state, \action^*, \state') \cdot \vect{x}(\state').
	\]
	As discussed above, for such state-action pair $(\state, \action^*)$, we have $(\state, \action^*) \notin D$ and thus $\action^* \in A \union A_\mathcal{H}$. Let $\varepsilon > 0$ such that
	\[
	\vect{x}(\state) = \varepsilon + \sum_{\state' \in \states} \delta(\state, \action^*, \state') \cdot \vect{x}(\state'),
	\]
	and define $\gamma = \varepsilon / (\sum_{\action \in \actions(\state)} \sum_{\state' \in \states} \mu(\action) \cdot \transRel(\state, \action, \state') \cdot \vect{x}(\state'))$.
	We then define the distribution $\mu' = \vect{1}_{a^*} \cdot (1 - \gamma) + \mu \cdot \gamma$, for which the following holds:
	\begin{align*}
		&\phantom{{}={}}\sum_{\action \in \actions(\state)} \sum_{\state' \in \states} \mu'(\action) \cdot \transRel(\state, \action, \state') \cdot \vect{x}(\state') \\
		&= \sum_{\action \in \actions(\state)} \sum_{\state' \in \states} (\vect{1}_{\action^*}(\action) \cdot (1-\gamma) + \mu(\action) \cdot \gamma) \cdot \transRel(\state, \action, \state') \cdot \vect{x}(\state')\\
		&=  \underbrace{\left((1-\gamma) \cdot \sum_{\state' \in \states} \transRel(\state, \action^*, \state') \cdot \vect{x}(\state') \right)}_{\leq \vect{x}(\state) - \varepsilon} + \underbrace{\left( \gamma \cdot \sum_{\action \in \actions(\state)} \sum_{\state' \in \states} \mu(\action)\cdot \transRel(\state, \action, \state') \cdot \vect{x}(\state') \right)}_{= \varepsilon} \\
		&\leq \vect{x}(\state)
	\end{align*}
	We can then set $\scheduler_1(\state) = \mu'$. By construction $\scheduler_1$ is then compliant with $T$ since it follows $\vect{x}$, never chooses a co-live (or unsafe) action and always plays live-actions with a (constant) positive probability.
	\qed
\end{proof}

\subsection{Relative Completeness}
\relativeCompleteness*
\begin{proof}
	\noindent We prove the first statement first. Suppose $\scheduler_1^*$ follows $T^*$. Then, we have that
	\[
	\inf_{\scheduler_2 \in \schedulers_2} \prob_{\sg, \init}^{\scheduler_1^*, \scheduler_2}(\Phi) \geq 1 - \vect{x}^*(\init) = \valueVector_{\sg, \Phi}(\init).
	\]
	Observe that the strategy template $\template$ from \cite{anand_permissive_templates_2023,phalakarn_winning_2024} \emph{does not} depend on the initial state $\init$. Hence, we can slightly generalize our results and obtain for all $\state \in \states$:
	\[
	\inf_{\scheduler_2 \in \schedulers_2} \prob_{\sg, \state}^{\scheduler_1^*, \scheduler_2}(\Phi) \geq 1 - \vect{x}^*(\state) = \valueVector_{\sg, \Phi}(\state).
	\]
	Hence, $\scheduler_1^*$ is a globally optimal strategy.
	
	\medskip
	
	\noindent Let us now prove the second statement. Assume that $\scheduler_1^*$ is a memoryless globally optimal strategy for $(\sg, \Phi)$. We need to show that $\scheduler_1^*$ follows $\vect{x}^*$ w.r.t. $I^*$. From \cite[Proposition 2.1]{chatterjee_quantitative_games_2004} we have that
	\begin{align}
		&\valueVector_{\sg, \Phi}(\state) \geq \sum_{\state' \in \states} \transRel(\state, \action, \state') \cdot \valueVector_{\sg, \Phi}(\state') & \forall \state \in \states_1,\; \forall \action \in \actions(\state),\\
		&\valueVector_{\sg, \Phi}(\state) \leq \sum_{\state' \in \states} \transRel(\state, \action, \state') \cdot \valueVector_{\sg, \Phi}(\state') & \forall \state \in \states_2,\; \forall \action \in \actions(\state).
	\end{align}
	This implies that for all $\state \in I^* \intersection \states_2$ and $\action \in \actions(\state)$ we have
	\begin{align*}
		\vect{x}^*(\state) = 1 - \valueVector_{\sg, \Phi}(\state) &\geq 1 - \sum_{\state' \in \states} \transRel(\state, \action, \state') \cdot \valueVector_{\sg, \Phi}(\state')\\
		&= \sum_{\state' \in \states} \transRel(\state, \action, \state') \cdot (1 - \valueVector_{\sg, \Phi}(\state'))\\
		&= \sum_{\state' \in \states} \transRel(\state, \action, \state') \cdot \vect{x}^*(\state').
	\end{align*}
	Hence \eqref{eq:opponent} is satisfied. Now suppose that \eqref{eq:strategy} is not satisfied. Then there exists a state $\state \in I^* \intersection \states_1$ such that
	\[
	\vect{x}^*(\state) < \sum_{\action \in \actions(\state)} \sum_{\state' \in \states} \scheduler_1^*(\state)(\action) \cdot \transRel(\state, \action, \state') \cdot \vect{x}^*(\state').
	\]
	This implies the following:
	\[
	\valueVector_{\sg, \Phi}(\state) > \sum_{\action \in \actions(\state)} \sum_{\state' \in \states} \scheduler_1^*(\state)(\action) \cdot \transRel(\state, \action, \state') \cdot \valueVector_{\sg, \Phi}(\state').
	\]
	Since $\Phi$ is a parity objective and thus prefix-independent, we get that
	\[
	\inf_{\scheduler_2 \in \schedulers_2}\prob_{\sg, \state}^{\scheduler_1^*, \scheduler_2}(\Phi) \leq \sum_{\action \in \actions(\state)} \sum_{\state' \in \states} \scheduler_1^*(\state)(\action) \cdot \transRel(\state, \action, \state') \cdot \valueVector_{\sg, \Phi}(\state')
	\]
	However, in total, this implies $\valueVector_{\sg, \Phi}(\state) > \inf_{\scheduler_2 \in \schedulers_2}\prob_{\sg, \state}^{\scheduler_1^*, \scheduler_2}(\Phi)$,
	contradicting the global optimality of $\scheduler_1^*$. Hence, \eqref{eq:strategy} is satisfied. Lastly, by definition of $I^*$ we have that for all $\state \in \states \setminus I^*$:
	\[
	\vect{x}^*(\state) = 1 - \valueVector_{\sg, \Phi}(\state) = 1.
	\]
	Hence $\scheduler_1^*$ follows $\vect{x}^*$. Finally, from \Cref{theorem:correctness-compliant-strat} and \Cref{lemma:conflict-free} we know that it is always possible to extract a strategy that follows $\stochasticTemplate^*$. From the first statement, we know that such a strategy is guaranteed to be globally optimal.
	\qed
\end{proof}

%% file: sections/appendix-experiments.tex
In this section we provide additional details concerning our implementation and experimental evaluation discussed in \Cref{section:experiments}. Our implementation and the considered models and properties are available at \url{https://doi.org/10.5281/zenodo.19929368} \cite{baier_2026_22012734}.

\subsection{Performance} \label{subsection:performance}
We now provide details concerning the performance of our proof-of-concept implementation. For the evaluation, we consider the following models (specified in the \textsc{Prism}-language):
\begin{itemize}
	\item \texttt{avoid} \cite{10.1007/978-3-030-53291-8_21} models a game in grid world, where an intruder tries to escape from an observer. We consider grid sizes $5 \times 5$, and $6 \times 6$. %
	\item \texttt{dice} \cite{kwiatkowska_prism_games_2020} represents a dice game between two players. We consider the model with $N \in \{10, 15, 20, 25\}$. %
	\item \texttt{hallway} \cite{10.1007/978-3-030-53291-8_21} models the interaction between a human and robot in a grid world. We consider the model with grid sizes $4 \times 4$ and $5 \times 5$.%
	\item \texttt{mdsm} \cite{ChenFKPS13} implements an energy management algorithm. We consider the model with $K \in \{8, 12, 16\}$. %
	\item \texttt{bigmec} \cite{KretinskyRSW22} is a scalable model that contains large maximal end components. We consider the model with $N \in \{5000, 10000, 15000\}$.
	\item \texttt{grid} \cite{BalsEKW24} is a grid world model. The original model from \cite{BalsEKW24} is an MDP. We turned the MDP into an SG by allowing the initial state to be chosen by the opponent player and consider grid sizes $16$, $32$, and $64$. 
\end{itemize}
Due to the prototypical nature of our implementation, the considered models are only of moderate size (between $257$ and $34645$ states).

\smallskip

\input{tables/reach.tex}
\input{tables/parity.tex}

\noindent\textbf{Reachability Objectives.} As an important special case, we first investigated the performance of our prototypical tool for reachability objectives. Specifically, we considered the models \texttt{avoid}, \texttt{dice}, \texttt{hallway}, and \texttt{mdsm} and reachability objectives that have been previously considered for these models \cite{kwiatkowska_prism_games_2020,10.1007/978-3-030-53291-8_21}. For \texttt{avoid} we consider the probability of eventually escaping. In the \texttt{dice} model we study the probability of a win for Player 1. In the \texttt{hallway} model the objective for the robot is to eventually reach the same location as the human. Lastly, in the \texttt{mdsm} model we consider the probability of reaching a specific configuration.

Our results are summarized in \Cref{tab:reachability}. For each case study we consider different configurations, resulting in different numbers of states $\card{\states}$ and state-action pairs $\card{\SA}$. The \emph{Strategy} column shows the time needed for \textsc{Prism-games} \cite{kwiatkowska_prism_games_2020} to synthesize an optimal strategy. The columns \emph{Invariant}, \emph{Certificate}, \emph{Template} show the times for computing the invariant, certificate and almost-sure strategy template, respectively. Finally, \emph{Total}\footnote{The sum of the runtimes of the individual steps may not exactly equal \emph{Total}, due to overhead not captured in the runtime of the individual steps.} describes the total time needed for synthesizing an \ASCert (excluding the time for synthesizing an initial strategy). It can be seen that \textsc{Prism-games} can compute a strategy within a few seconds. The runtimes for the synthesis of our \ASCerts range roughly from $3$ to $50$ seconds, where the computation of the certificate and template are the most time-consuming parts.

\smallskip

\noindent\textbf{Parity Objectives.} We considered parity objectives for the \texttt{bigmec}, \texttt{grid} and \texttt{hallway} models. For \texttt{bigmec} we hand-crafted two different parity objectives with $3$ and $5$ colors, respectively. For the \texttt{grid} model we took a property from \cite{BalsEKW24} and converted it to a parity objective with $3$ colors. In the \texttt{hallway} case study, two parity objectives with $3$ and $5$ colors are considered. The first property expresses that the human and robot share the same location infinitely often. The second property additionally expresses that both share the location infinitely often also when the robot is not damaged.

The results are summarized in \Cref{tab:parity-small} in the main body. For convenience, we have also included the table here in \Cref{tab:parity}. Recall that for parity objectives we use \textsc{Mungojerrie} to \emph{learn} a strategy. For the learning we needed to tune hyperparameters. In particular, the runtime of \textsc{Mungojerrie} depends on the choice of the hyperparameters. Thus, we omit the runtime for \textsc{Mungojerrie}. In the column \emph{\#Colors}, each value corresponds to one parity objective and indicates the number of colors of this parity objective. As before, the computation of certificates and templates is the most expensive part in the synthesis of our \ASCerts. We observe that the number of colors in the objective has an influence on the computation time for the template. The other steps seem mostly unaffected.

\smallskip

\noindent\textbf{Validation.} We also validated all synthesized \ASCerts (for both reachability and parity objectives). To this end, we first extracted a memoryless strategy $\scheduler_1$ from the \ASCert (as explained in \Cref{sec:extraction}) and subsequently computed the satisfaction probability in the induced MDP $\sg^\scheduler_1$ via \textsc{Storm} \cite{HenselJKQV22}. For all computed \ASCerts, we could confirm that the satisfaction probability is above the threshold $\lambda$ guaranteed by the \ASCert.

\smallskip

\noindent\textbf{Discussion.} In the synthesis of our \ASCerts, the computation of the certificate and almost surely winning strategy template are the most time-consuming steps. Recall that for the computation of the certificate we solve a linear program (in the size of the game). Upon inspection, we observed that the construction of the linear program consumes more than $80\%$ of the time needed for computing the certificate on average. For the computation of the template we construct a non-stochastic parity game using the gadgets from \cite{chatterjee_quantitative_games_2004}. The construction of these games requires around $60\%$ of the time (on average) needed for computing the template. Since the constructions of the linear program and non-stochastic parity game are currently implemented in Python, we see a lot of room for improvement when switching to a faster language.

Overall, we conclude that for the considered (moderately-sized) models our proof-of-concept can synthesize \ASCerts in a reasonable amount of time.

\subsection{Runtime Adaptation} \label{subsection:adaptation}
We now provide additional details concerning the runtime adaptation case study. Firstly, as mentioned in the main body, our model is based on the \texttt{hallway} model from \cite{10.1007/978-3-030-53291-8_21}. Our modified variant has $19248$ states and $65008$ state-action pairs.

\smallskip
\begin{figure}[t!]
	\centering
	
	\begin{subfigure}[b]{0.48\textwidth}
		\centering
		\includegraphics[width=\textwidth]{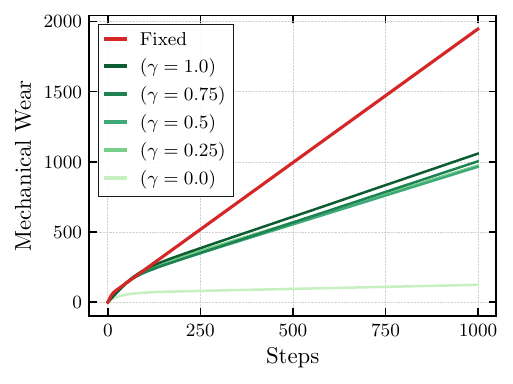}
		\caption{Permissiveness maximization in all states.}
		\label{subfig:permissiveness-all}
	\end{subfigure}
	\hfill
	\begin{subfigure}[b]{0.48\textwidth}
		\centering
		\includegraphics[width=\textwidth]{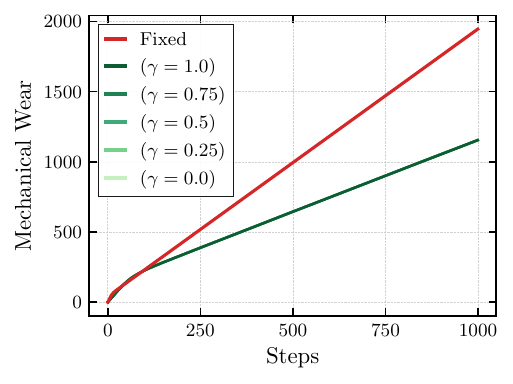}
		\caption{No permissiveness maximization.\\ \phantom{hello}}
		\label{subfig:permissiveness-none}
	\end{subfigure}
	
	\caption{Average mechanical wear for different configurations.}
	\label{fig:permissiveness-different-configurations}
\end{figure}

\noindent\textbf{Synthesis of Permissive Certificates.} Recall the linear program from \Cref{subsection:permissive-lp} for synthesizing certificates which maximize the overall permissiveness, where we can focus the maximization on a subset of states. In our experiments we computed \ASCerts $\stochasticTemplate_\gamma = (I, \vect{x}, \template)$ where $\vect{x}$ maximizes the overall permissiveness w.r.t. the following set of states:
\begin{description}
	\item[A:] In grid locations around the first puddle {\puddle}.
	\item[B:] In all grid locations.
\end{description}
The results for \textbf{A} have been presented in \cref{fig:runtime-adaptation} in \Cref{section:experiments} and enabled the robot to effectively reduce the mechanical wear at runtime. For \textbf{B}, the results are shown in \cref{subfig:permissiveness-all}. It can be seen that for all thresholds the mechanical wear can be reduced noticeably. In fact, the mechanical wear can be reduced even more compared to \textbf{A} in \cref{fig:runtime-adaptation} for $\gamma \neq 0.25$. However, for $\gamma = 0.25$, the reduction in mechanical wear is slightly less than in \textbf{A}. Overall, this shows that the flexibility provided by our synthesis procedure depends on the choice of states in which the permissiveness is maximized.

Lastly, we also synthesized \ASCerts where we did not maximize permissiveness at all. Specifically, we solved the LP in \Cref{subsection:permissive-lp} with zero as objective function, thereby allowing \textsc{Gurobi} to return any feasible solution. The results are depicted in \Cref{subfig:permissiveness-none} and show that for all probability thresholds the same reduction in mechanical wear can be achieved. In particular, unlike before, the higher failure tolerances (i.e., lower values for $\gamma$) are not leveraged to provide more flexibility to the robot. This highlights the usefulness of our permissiveness maximization heuristic.

%% file: tables/reach.tex
\begin{table}[t]
	\centering
	\setlength{\tabcolsep}{1.4pt}  %
	\caption{Runtimes for reachability objectives. All runtimes are wall-clock times and specified in seconds.}
	\label{tab:reachability}
	\begin{tabular}{l@{\hspace{2.3em}}cc@{\hspace{2.3em}}r@{\hspace{2.3em}}rrrrr}
		\toprule
		& & & & \multicolumn{4}{c}{\textsf{ASCert} Synthesis} \\
		Case Study & $\card{\states}$ & $\card{\SA}$ & Strategy & Invariant & Certificate & Template & Total \\
		\midrule
\multirow{2}{*}{\texttt{avoid}}& 8584 & 20270 & 1.594 & 0.919 & 7.151 & 1.902 & 10.082 \\
& 16748 & 39354 & 1.607 & 3.297 & 18.341 & 4.577 & 26.249 \\
\midrule
\multirow{4}{*}{\texttt{dice}}& 5755 & 7429 & 1.597 & 0.459 & 1.944 & 0.728 & 3.183 \\
& 12685 & 16549 & 1.584 & 1.961 & 4.139 & 2.613 & 8.730 \\
& 22315 & 29269 & 1.620 & 5.954 & 7.160 & 6.884 & 20.016 \\
& 34645 & 45589 & 1.639 & 14.184 & 11.661 & 15.563 & 41.434 \\
\midrule
\multirow{2}{*}{\texttt{hallway}}& 10240 & 26624 & 1.939 & 1.291 & 9.854 & 3.078 & 14.224 \\
& 25000 & 65000 & 4.217 & 7.165 & 24.558 & 11.605 & 43.361 \\
\midrule
\multirow{3}{*}{\texttt{mdsm}}& 16248 & 22416 & 1.612 & 3.055 & 6.892 & 5.713 & 15.733 \\
& 24888 & 34368 & 1.746 & 6.944 & 10.435 & 12.029 & 29.448 \\
& 33528 & 46320 & 1.646 & 12.268 & 13.645 & 20.645 & 46.607 \\
\bottomrule
	\end{tabular}
\end{table}

%% file: tables/parity.tex
\begin{table}[t]
	\centering
	\setlength{\tabcolsep}{1.4pt}  %
	\caption{Runtimes for parity objectives. All runtimes are wall-clock times and specified in seconds.}
	\label{tab:parity}
	\begin{tabular}{l @{\hspace{3em}} ccc @{\hspace{3em}} rrrrr}
		\toprule
		& & & & \multicolumn{4}{c}{\textsf{ASCert} Synthesis} \\
		Case Study & \#Colors & $\card{\states}$ & $\card{\SA}$ & Invariant & Certificate & Template & Total \\
\midrule
\multirow{9}{*}{\texttt{bigmec}}& \multirow{3}{*}{3} & 10003 & 21432 & 1.133 & 11.378 & 1.613 & 14.303 \\
&  & 20003 & 42860 & 4.301 & 28.648 & 4.297 & 37.351 \\
&  & 30003 & 64290 & 9.740 & 49.046 & 8.155 & 67.077 \\
\cmidrule{2-8}
& \multirow{3}{*}{5} & 10003 & 21432 & 1.217 & 11.398 & 1.650 & 14.269 \\
&  & 20003 & 42860 & 4.566 & 28.390 & 4.467 & 37.430 \\
&  & 30003 & 64290 & 10.128 & 48.760 & 8.559 & 67.459 \\
\midrule
\multirow{3}{*}{\texttt{grid}}& \multirow{3}{*}{3} & 257 & 1029 & 0.011 & 0.577 & 0.129 & 0.949 \\
&  & 1025 & 4100 & 0.051 & 2.307 & 0.490 & 2.850 \\
&  & 4097 & 16388 & 0.340 & 9.228 & 2.193 & 11.767 \\
\midrule
\multirow{4}{*}{\texttt{hallway}}& \multirow{2}{*}{3} & 10240 & 26624 & 1.308 & 9.978 & 2.553 & 13.839 \\
&  & 25000 & 65000 & 6.937 & 24.857 & 9.291 & 41.136 \\
\cmidrule{2-8}
& \multirow{2}{*}{5} & 10240 & 26624 & 1.331 & 9.936 & 3.687 & 14.962 \\
&  & 25000 & 65000 & 7.166 & 24.941 & 13.363 & 45.488 \\
\bottomrule
\end{tabular}
\end{table}

%% file: main.bbl
\begin{thebibliography}{10}
\providecommand{\url}[1]{\texttt{#1}}
\providecommand{\urlprefix}{URL }
\providecommand{\doi}[1]{https://doi.org/#1}

\bibitem{abate_quantitative_2025}
Abate, A., Giacobbe, M., Roy, D.: Quantitative supermartingale certificates.
  In: Piskac, R., Rakamaric, Z. (eds.) Computer Aided Verification. Lecture
  Notes in Computer Science, vol. 15932, pp. 3--28. Springer (2025).
  \doi{10.1007/978-3-031-98679-6\_1}

\bibitem{Alfaro97}
de~Alfaro, L.: Formal verification of probabilistic systems. Ph.D. thesis,
  Stanford University, {USA} (1997),
  \url{https://searchworks.stanford.edu/view/3910936}

\bibitem{AlshiekhBEKNT18}
Alshiekh, M., Bloem, R., Ehlers, R., K{\"{o}}nighofer, B., Niekum, S., Topcu,
  U.: Safe reinforcement learning via shielding. In: McIlraith, S.A.,
  Weinberger, K.Q. (eds.) Proceedings of the Thirty-Second {AAAI} Conference on
  Artificial Intelligence, (AAAI-18), the 30th innovative Applications of
  Artificial Intelligence (IAAI-18), and the 8th {AAAI} Symposium on
  Educational Advances in Artificial Intelligence (EAAI-18), New Orleans,
  Louisiana, USA, February 2-7, 2018. pp. 2669--2678. {AAAI} Press (2018).
  \doi{10.1609/AAAI.V32I1.11797}

\bibitem{concurrent_templates2026}
Anand, A., Baier, C., Chau, C., Kl{\"{u}}ppelholz, S., Mirzaei, A., Nayak,
  S.P., Schmuck, A.: Concurrent permissive strategy templates. In: Junges, S.,
  Katz, G. (eds.) Tools and Algorithms for the Construction and Analysis of
  Systems - 32nd International Conference, {TACAS} 2026, Held as Part of the
  International Joint Conferences on Theory and Practice of Software, {ETAPS}
  2026, Turin, Italy, April 11-16, 2026, Proceedings, Part {I}. pp. 377--397.
  Lecture Notes in Computer Science, Springer (2026).
  \doi{10.1007/978-3-032-22752-2\_20}

\bibitem{anand_computing_2023}
Anand, A., Mallik, K., Nayak, S.P., Schmuck, A.K.: Computing adequately
  permissive assumptions for synthesis. In: Tools and {{Algorithms}} for the
  {{Construction}} and {{Analysis}} of {{Systems}}. Lecture Notes in Computer
  Science, vol. 13994, pp. 211--228. Springer, Cham (2023).
  \doi{10.1007/978-3-031-30820-8\_15}

\bibitem{anand_quantitative_2025}
Anand, A., Nayak, S.P., Raha, R., Sa{\u g}lam, I., Schmuck, A.: Quantitative
  strategy templates. In: D'Souza, M., Komondoor, R., Srivathsan, B. (eds.)
  Automated Technology for Verification and Analysis - 23rd International
  Symposium, {ATVA} 2025, Bengaluru, India, October 27-31, 2025, Proceedings.
  Lecture Notes in Computer Science, vol. 16145, pp. 65--86. Springer (2025).
  \doi{10.1007/978-3-032-08707-2\_4}

\bibitem{anand2025followstarsdynamicomegaregular}
Anand, A., Nayak, S.P., Raha, R., Schmuck, A.K.: Follow the stars: Dynamic
  $\omega$-regular shielding of learned probabilistic policies. In: Proceedings
  of the 25th International Conference on Autonomous Agents and Multiagent
  Systems. pp. 2978–--2986. AAMAS '26, International Foundation for
  Autonomous Agents and Multiagent Systems, Richland, SC (2026).
  \doi{10.65109/WLFP2035}

\bibitem{anand_permissive_templates_2023}
Anand, A., Nayak, S.P., Schmuck, A.: Synthesizing permissive winning strategy
  templates for parity games. In: Enea, C., Lal, A. (eds.) Computer Aided
  Verification - 35th International Conference, {CAV} 2023, Paris, France, July
  17-22, 2023, Proceedings, Part {I}. Lecture Notes in Computer Science, vol.
  13964, pp. 436--458. Springer (2023). \doi{10.1007/978-3-031-37706-8\_22}

\bibitem{anand_templates_2024}
Anand, A., Nayak, S.P., Schmuck, A.: Strategy templates - robust certified
  interfaces for interacting systems. In: Akshay, S., Niemetz, A.,
  Sankaranarayanan, S. (eds.) Automated Technology for Verification and
  Analysis - 22nd International Symposium, {ATVA} 2024, Kyoto, Japan, October
  21-25, 2024, Proceedings, Part {I}. Lecture Notes in Computer Science, vol.
  15054, pp. 22--41. Springer (2024). \doi{10.1007/978-3-031-78709-6\_2}

\bibitem{qcomp2024}
Andriushchenko, R., Bork, A., Budde, C.E., Ceska, M., Grover, K., Hahn, E.M.,
  Hartmanns, A., Israelsen, B., Jansen, N., Jeppson, J., Junges, S.,
  K{\"{o}}hl, M.A., K{\"{o}}nighofer, B., Kret{\'{\i}}nsk{\'{y}}, J.,
  Meggendorfer, T., Parker, D., Pranger, S., Quatmann, T., Ruijters, E.,
  Taylor, L., Volk, M., Weininger, M., Zhang, Z.: Tools at the frontiers of
  quantitative verification. CoRR  \textbf{abs/2405.13583} (2024).
  \doi{10.48550/ARXIV.2405.13583}

\bibitem{andriushchenko2026umbunifiedmarkovbinary}
Andriushchenko, R., Hartmanns, A., Jeppson, J., Junges, S., Meggendorfer, T.,
  Parker, D., Quatmann, T., Weininger, M.: {UMB}: A unified {Markov} binary
  format for probabilistic model checking (extended version) (2026),
  \url{https://arxiv.org/abs/2606.17811}

\bibitem{baier_2026_22012734}
Baier, C., Cauquil, D., Chau, C., Klüppelholz, S., Schmuck, A.K.:
  Supplementary material for {ATVA} 2026 (2026). \doi{10.5281/zenodo.22012734}

\bibitem{baier_principles_2008}
Baier, C., Katoen, J.: Principles of Model Checking. The MIT Press, Cambridge,
  MA (2008)

\bibitem{BalsEKW24}
Bals, S., Evangelidis, A., Kret{\'{\i}}nsk{\'{y}}, J., Waibel, J.: {MULTIGAIN}
  2.0: {MDP} controller synthesis for multiple mean-payoff, {LTL} and
  steady-state constraints. In: {\'{A}}brah{\'{a}}m, E., Jr., M.M. (eds.)
  Proceedings of the 27th {ACM} International Conference on Hybrid Systems:
  Computation and Control, {HSCC} 2024, Hong Kong SAR, China, May 14-16, 2024.
  pp. 24:1--24:7. {ACM} (2024). \doi{10.1145/3641513.3650135}

\bibitem{belta2017formal}
Belta, C., Yordanov, B., Aydin~Gol, E.: Formal methods for discrete-time
  dynamical systems. Springer (2017)

\bibitem{bloem_shield_2015}
Bloem, R., K{\"o}nighofer, B., K{\"o}nighofer, R., Wang, C.: Shield synthesis:
  {R}untime enforcement for reactive systems. In: Tools and {{Algorithms}} for
  the {{Construction}} and {{Analysis}} of {{Systems}}. pp. 533--548. Lecture
  {{Notes}} in {{Computer Science}}, Springer, Berlin, Heidelberg (2015)

\bibitem{ChakarovS13}
Chakarov, A., Sankaranarayanan, S.: Probabilistic program analysis with
  martingales. In: Sharygina, N., Veith, H. (eds.) Computer Aided Verification
  - 25th International Conference, {CAV} 2013, Saint Petersburg, Russia, July
  13-19, 2013. Proceedings. Lecture Notes in Computer Science, vol.~8044, pp.
  511--526. Springer (2013). \doi{10.1007/978-3-642-39799-8\_34}

\bibitem{ChatterjeeGMZ22}
Chatterjee, K., Goharshady, A.K., Meggendorfer, T., Zikelic, D.: Sound and
  complete certificates for quantitative termination analysis of probabilistic
  programs. In: Shoham, S., Vizel, Y. (eds.) Computer Aided Verification - 34th
  International Conference, {CAV} 2022, Haifa, Israel, August 7-10, 2022,
  Proceedings, Part {I}. Lecture Notes in Computer Science, vol. 13371, pp.
  55--78. Springer (2022). \doi{10.1007/978-3-031-13185-1\_4}

\bibitem{ChatterjeeH06}
Chatterjee, K., Henzinger, T.A.: Strategy improvement and randomized
  subexponential algorithms for stochastic parity games. In: Durand, B.,
  Thomas, W. (eds.) {STACS} 2006, 23rd Annual Symposium on Theoretical Aspects
  of Computer Science, Marseille, France, February 23-25, 2006, Proceedings.
  Lecture Notes in Computer Science, vol.~3884, pp. 512--523. Springer (2006).
  \doi{10.1007/11672142\_42}

\bibitem{chatterjee_gist_2010}
Chatterjee, K., Henzinger, T.A., Jobstmann, B., Radhakrishna, A.: Gist: {A}
  solver for probabilistic games. In: Touili, T., Cook, B., Jackson, P.B.
  (eds.) Computer Aided Verification, 22nd International Conference, {CAV}
  2010, Edinburgh, UK, July 15-19, 2010. Proceedings. Lecture Notes in Computer
  Science, vol.~6174, pp. 665--669. Springer (2010).
  \doi{10.1007/978-3-642-14295-6\_57}

\bibitem{chatterjee_quantitative_games_2004}
Chatterjee, K., Jurdzi\'nski, M., Henzinger, T.A.: Quantitative stochastic
  parity games. In: Munro, J.I. (ed.) Proceedings of the Fifteenth Annual
  {ACM-SIAM} Symposium on Discrete Algorithms, {SODA} 2004, New Orleans,
  Louisiana, USA, January 11-14, 2004. pp. 121--130. {SIAM} (2004),
  \url{http://dl.acm.org/citation.cfm?id=982792.982808}

\bibitem{10.1007/978-3-030-53291-8_21}
Chatterjee, K., Katoen, J., Weininger, M., Winkler, T.: Stochastic games with
  lexicographic reachability-safety objectives. In: Lahiri, S.K., Wang, C.
  (eds.) Computer Aided Verification - 32nd International Conference, {CAV}
  2020, Los Angeles, CA, USA, July 21-24, 2020, Proceedings, Part {II}. pp.
  398--420. Lecture Notes in Computer Science, Springer (2020).
  \doi{10.1007/978-3-030-53291-8\_21}

\bibitem{ChatterjeeNZ17}
Chatterjee, K., Novotn{\'{y}}, P., Zikelic, D.: Stochastic invariants for
  probabilistic termination. In: Castagna, G., Gordon, A.D. (eds.) Proceedings
  of the 44th {ACM} {SIGPLAN} Symposium on Principles of Programming Languages,
  {POPL} 2017, Paris, France, January 18-20, 2017. pp. 145--160. {ACM} (2017).
  \doi{10.1145/3009837.3009873}

\bibitem{chatterjee_fixed_point_2025}
Chatterjee, K., Quatmann, T., Sch{\"{a}}ffeler, M., Weininger, M., Winkler, T.,
  Zilken, D.: Fixed point certificates for reachability and expected rewards in
  mdps. In: Gurfinkel, A., Heule, M. (eds.) Tools and Algorithms for the
  Construction and Analysis of Systems - 31st International Conference, {TACAS}
  2025, Held as Part of the International Joint Conferences on Theory and
  Practice of Software, {ETAPS} 2025, Hamilton, ON, Canada, May 3-8, 2025,
  Proceedings, Part {II}. Lecture Notes in Computer Science, vol. 15697, pp.
  130--151. Springer (2025). \doi{10.1007/978-3-031-90653-4\_7}

\bibitem{ChenFKPS13}
Chen, T., Forejt, V., Kwiatkowska, M.Z., Parker, D., Simaitis, A.: Automatic
  verification of competitive stochastic systems. Formal Methods Syst. Des.
  \textbf{43}(1),  61--92 (2013). \doi{10.1007/S10703-013-0183-7}

\bibitem{DBLP:conf/tacas/DragerFKPU14}
Dr{\"{a}}ger, K., Forejt, V., Kwiatkowska, M.Z., Parker, D., Ujma, M.:
  Permissive controller synthesis for probabilistic systems. In:
  {\'{A}}brah{\'{a}}m, E., Havelund, K. (eds.) Tools and Algorithms for the
  Construction and Analysis of Systems - 20th International Conference, {TACAS}
  2014, Held as Part of the European Joint Conferences on Theory and Practice
  of Software, {ETAPS} 2014, Grenoble, France, April 5-13, 2014. Proceedings.
  Lecture Notes in Computer Science, vol.~8413, pp. 531--546. Springer (2014).
  \doi{10.1007/978-3-642-54862-8\_44}

\bibitem{Durret2019}
Durrett, R.: Probability: Theory and Examples, 5th Edition. Cambridge
  University Press (2019). \doi{10.1017/9781108591034}

\bibitem{fijalkow2025gamesgraphslogicautomata}
Fijalkow, N., Aiswarya, C., Avni, G., Bertrand, N., Bouyer, P., Brenguier, R.,
  Carayol, A., Casares, A., Fearnley, J., Gastin, P., Gimbert, H., Henzinger,
  T.A., Horn, F., Ibsen-Jensen, R., Markey, N., Monmege, B., Novotný, P.,
  Ohlmann, P., Randour, M., Sankur, O., Schmitz, S., Serre, O., Skomra, M.,
  Sznajder, N., Vandenhove, P.: Games on graphs: From logic and automata to
  algorithms (2025), \url{https://arxiv.org/abs/2305.10546}

\bibitem{finkbeiner2016synthesis}
Finkbeiner, B.: Synthesis of reactive systems. In: Dependable Software Systems
  Engineering, pp. 72--98. IOS Press (2016)

\bibitem{funke_farkas_2020}
Funke, F., Jantsch, S., Baier, C.: Farkas certificates and minimal witnesses
  for probabilistic reachability constraints. In: {{TACAS}}. {{LNCS}}, vol.
  12078, pp. 324--345. Springer International Publishing, Cham (2020).
  \doi{10.1007/978-3-030-45190-5_18}

\bibitem{gurobi}
{Gurobi Optimization, LLC}: {Gurobi Optimizer Reference Manual} (2026),
  \url{https://www.gurobi.com}

\bibitem{hahn_mungojerrie_2023}
Hahn, E.M., Perez, M., Schewe, S., Somenzi, F., Trivedi, A., Wojtczak, D.:
  Mungojerrie: Linear-time objectives in model-free reinforcement learning. In:
  Sankaranarayanan, S., Sharygina, N. (eds.) Tools and Algorithms for the
  Construction and Analysis of Systems - 29th International Conference, {TACAS}
  2023, Held as Part of the European Joint Conferences on Theory and Practice
  of Software, {ETAPS} 2023, Paris, France, April 22-27, 2023, Proceedings,
  Part {I}. Lecture Notes in Computer Science, vol. 13993, pp. 527--545.
  Springer (2023). \doi{10.1007/978-3-031-30823-9\_27}

\bibitem{havelund2002synthesizing}
Havelund, K., Ro{\c{s}}u, G.: Synthesizing monitors for safety properties. In:
  Tools and Algorithms for the Construction and Analysis of Systems: 8th
  International Conference, TACAS 2002 Held as Part of the Joint European
  Conferences on Theory and Practice of Software, ETAPS 2002, Grenoble, France,
  April 8--12, 2002. pp. 342--356. Springer (2002)

\bibitem{HenselJKQV22}
Hensel, C., Junges, S., Katoen, J., Quatmann, T., Volk, M.: The probabilistic
  model checker storm. Int. J. Softw. Tools Technol. Transf.  \textbf{24}(4),
  589--610 (2022). \doi{10.1007/S10009-021-00633-Z}

\bibitem{HenzingerMSZ25}
Henzinger, T.A., Mallik, K., Sadeghi, P., Zikelic, D.: Supermartingale
  certificates for quantitative omega-regular verification and control. In:
  Piskac, R., Rakamaric, Z. (eds.) Computer Aided Verification - 37th
  International Conference, {CAV} 2025, Zagreb, Croatia, July 23-25, 2025,
  Proceedings, Part {II}. Lecture Notes in Computer Science, vol. 15932, pp.
  29--55. Springer (2025). \doi{10.1007/978-3-031-98679-6\_2}

\bibitem{DBLP:journals/corr/abs-2510-03481}
Huynh, K.V., Parker, D., Feng, L.: Robust permissive controller synthesis for
  interval mdps. CoRR  \textbf{abs/2510.03481} (2025).
  \doi{10.48550/ARXIV.2510.03481}

\bibitem{jantsch_dissertation_2022}
Jantsch, S.: Certificates and Witnesses for Probabilistic Model Checking. Ph.D.
  thesis, Dresden University of Technology, Germany (2022),
  \url{https://nbn-resolving.org/urn:nbn:de:bsz:14-qucosa2-804732}

\bibitem{konighofer2022correct}
K{\"o}nighofer, B., Bloem, R., Ehlers, R., Pek, C.: Correct-by-construction
  runtime enforcement in {AI}--a survey. In: Principles of Systems Design:
  Essays Dedicated to Thomas A. Henzinger on the Occasion of His 60th Birthday,
  pp. 650--663. Springer (2022)

\bibitem{KressGazitFainekosPappas2009}
Kress-Gazit, H., Fainekos, G.E., Pappas, G.J.: Temporal-logic-based reactive
  mission and motion planning. IEEE Transactions on Robotics  \textbf{25}(6),
  1370--1381 (2009). \doi{10.1109/TRO.2009.2030225}

\bibitem{KretinskyM19}
Kret{\'{\i}}nsk{\'{y}}, J., Meggendorfer, T.: Of cores: {A} partial-exploration
  framework for markov decision processes. In: Fokkink, W.J., van Glabbeek, R.
  (eds.) 30th International Conference on Concurrency Theory, {CONCUR} 2019,
  Amsterdam, The Netherlands, August 27-30, 2019. LIPIcs, vol.~140, pp.
  5:1--5:17. Schloss Dagstuhl - Leibniz-Zentrum f{\"{u}}r Informatik (2019).
  \doi{10.4230/LIPICS.CONCUR.2019.5}

\bibitem{KretinskyRSW22}
Kret{\'{\i}}nsk{\'{y}}, J., Ramneantu, E., Slivinskiy, A., Weininger, M.:
  Comparison of algorithms for simple stochastic games. Inf. Comput.
  \textbf{289}(Part),  104885 (2022). \doi{10.1016/J.IC.2022.104885}

\bibitem{KNP11}
Kwiatkowska, M., Norman, G., Parker, D.: {PRISM} 4.0: Verification of
  probabilistic real-time systems. In: Gopalakrishnan, G., Qadeer, S. (eds.)
  Proc. 23rd International Conference on Computer Aided Verification (CAV'11).
  LNCS, vol.~6806, pp. 585--591. Springer (2011)

\bibitem{kwiatkowska_prism_games_2020}
Kwiatkowska, M., Norman, G., Parker, D., Santos, G.: Prism-games 3.0:
  Stochastic game verification with concurrency, equilibria and time. In:
  Lahiri, S.K., Wang, C. (eds.) Computer Aided Verification - 32nd
  International Conference, {CAV} 2020, Los Angeles, CA, USA, July 21-24, 2020,
  Proceedings, Part {II}. Lecture Notes in Computer Science, vol. 12225, pp.
  475--487. Springer (2020). \doi{10.1007/978-3-030-53291-8\_25}

\bibitem{KwiatkowskaNP12}
Kwiatkowska, M.Z., Norman, G., Parker, D.: The {PRISM} benchmark suite. In:
  Ninth International Conference on Quantitative Evaluation of Systems, {QEST}
  2012, London, United Kingdom, September 17-20, 2012. pp. 203--204. {IEEE}
  Computer Society (2012). \doi{10.1109/QEST.2012.14}

\bibitem{lee1999runtime}
Lee, I., Kannan, S., Kim, M., Sokolsky, O., Viswanathan, M.: Runtime assurance
  based on formal specifications. In: International Conference on Parallel and
  Distributed Processing Techniques and Applications. pp. 279--287. PDPTA
  (1999)

\bibitem{lindemann2025formal}
Lindemann, L., Dimarogonas, D.V.: Formal Methods for Multi-Agent Feedback
  Control Systems. MIT Press (2025)

\bibitem{MajumdarS25}
Majumdar, R., Sathiyanarayana, V.R.: Sound and complete proof rules for
  probabilistic termination. Proc. {ACM} Program. Lang.  \textbf{9}({POPL}),
  1871--1902 (2025). \doi{10.1145/3704899}

\bibitem{majumdar_supermatingale_2024}
Majumdar, R., Sathiyanarayana, V.R., Soudjani, S.: Necessary and sufficient
  certificates for almost sure reachability. {IEEE} Control. Syst. Lett.
  \textbf{8},  2703--2708 (2024). \doi{10.1109/LCSYS.2024.3507279}

\bibitem{certifying_algorithms_2011}
McConnell, R.M., Mehlhorn, K., N{\"{a}}her, S., Schweitzer, P.: Certifying
  algorithms. Comput. Sci. Rev.  \textbf{5}(2),  119--161 (2011).
  \doi{10.1016/J.COSREV.2010.09.009}

\bibitem{AlthoffBeltaReviewFMCEforAutonomousDriving}
Mehdipour, N., Althoff, M., Tebbens, R.D., Belta, C.: Formal methods to comply
  with rules of the road in autonomous driving: State of the art and grand
  challenges. Automatica  \textbf{152},  110692 (2023).
  \doi{10.1016/j.automatica.2022.110692}

\bibitem{10221705}
Nayak, S.P., Egidio, L.N., Della~Rossa, M., Schmuck, A.K., Jungers, R.M.:
  Context-triggered abstraction-based control design. IEEE Open Journal of
  Control Systems  \textbf{2},  277--296 (2023).
  \doi{10.1109/OJCSYS.2023.3305835}

\bibitem{phalakarn_winning_2024}
Phalakarn, K., Pruekprasert, S., Hasuo, I.: Winning strategy templates for
  stochastic parity games towards permissive and resilient control. In:
  Anutariya, C., Bonsangue, M.M. (eds.) Theoretical Aspects of Computing -
  {ICTAC} 2024 - 21st International Colloquium, Bangkok, Thailand, November
  25-29, 2024, Proceedings. Lecture Notes in Computer Science, vol. 15373, pp.
  197--214. Springer (2024). \doi{10.1007/978-3-031-77019-7\_12}

\bibitem{phalakarn_templates_2025}
Phalakarn, K., Pruekprasert, S., Hasuo, I.: Strategy templates for almost-sure
  and positive winning of stochastic parity games towards permissive and
  resilient control. Theor. Comput. Sci.  \textbf{1057},  115535 (2025).
  \doi{10.1016/J.TCS.2025.115535}

\bibitem{ramadge1989control}
Ramadge, P.J., Wonham, W.M.: The control of discrete event systems. Proceedings
  of the IEEE  \textbf{77}(1),  81--98 (1989)

\bibitem{seshia2022toward}
Seshia, S.A., Sadigh, D., Sastry, S.S.: Toward verified artificial
  intelligence. Communications of the ACM  \textbf{65}(7),  46--55 (2022)

\bibitem{tabuada-2009-verification}
Tabuada, P.: Verification and control of hybrid systems. Springer, New York,
  NY, 2009 edn. (June 2009)

\bibitem{TakisakaOUH18}
Takisaka, T., Oyabu, Y., Urabe, N., Hasuo, I.: Ranking and repulsing
  supermartingales for reachability in probabilistic programs. In: Lahiri,
  S.K., Wang, C. (eds.) Automated Technology for Verification and Analysis -
  16th International Symposium, {ATVA} 2018, Los Angeles, CA, USA, October
  7-10, 2018, Proceedings. Lecture Notes in Computer Science, vol. 11138, pp.
  476--493. Springer (2018). \doi{10.1007/978-3-030-01090-4\_28}

\bibitem{TakisakaOUH21}
Takisaka, T., Oyabu, Y., Urabe, N., Hasuo, I.: Ranking and repulsing
  supermartingales for reachability in randomized programs. {ACM} Trans.
  Program. Lang. Syst.  \textbf{43}(2),  5:1--5:46 (2021).
  \doi{10.1145/3450967}

\bibitem{YIN2024100940}
Yin, X., Gao, B., Yu, X.: Formal synthesis of controllers for safety-critical
  autonomous systems: Developments and challenges. Annual Reviews in Control
  \textbf{57},  100940 (2024). \doi{10.1016/j.arcontrol.2024.100940}

\bibitem{zikelic_rasm_2023}
{\v Z}ikeli\'c, {\DJ}., Lechner, M., Henzinger, T.A., Chatterjee, K.: Learning
  control policies for stochastic systems with reach-avoid guarantees. In:
  Williams, B., Chen, Y., Neville, J. (eds.) Thirty-Seventh {AAAI} Conference
  on Artificial Intelligence, {AAAI} 2023, Thirty-Fifth Conference on
  Innovative Applications of Artificial Intelligence, {IAAI} 2023, Thirteenth
  Symposium on Educational Advances in Artificial Intelligence, {EAAI} 2023,
  Washington, DC, USA, February 7-14, 2023. pp. 11926--11935. {AAAI} Press
  (2023). \doi{10.1609/AAAI.V37I10.26407}

\end{thebibliography}
